%% file: paper.tex
\documentclass[10pt,letterpaper]{article}

\pdfoutput=1

\usepackage[margin=1in]{geometry}
\usepackage[utf8]{inputenc}
\usepackage[T1]{fontenc}
\usepackage[nopatch=footnote]{microtype}
\usepackage{amsfonts}
\usepackage{amssymb}
\usepackage{amsmath}
\usepackage{amsthm}
\usepackage{booktabs}
\usepackage[inline]{enumitem}
\usepackage{setspace}
\usepackage[ruled,noend,linesnumbered]{algorithm2e}
\usepackage[hidelinks,bookmarks,bookmarksopen,bookmarksnumbered,bookmarksdepth=subsection]{hyperref}
\usepackage[capitalize]{cleveref}
\usepackage{mathtools}
\usepackage{tikz}
\usetikzlibrary{decorations.pathreplacing}
\usetikzlibrary{matrix}
\usepackage[font=small,labelfont=bf]{caption}
\usepackage{titlesec}
\usepackage{thmtools}
\usepackage{tablefootnote}
\usepackage{soul}
\usepackage{lineno}
\usepackage{framed}
\usepackage{thm-restate}
\usepackage{calc}
\usetikzlibrary{decorations.pathreplacing}
\usetikzlibrary{patterns}

\input{settings}
\input{macros}

\begin{document}

\title{Hardness of Frequency-Related Queries on Compressed Strings}

\author{
  \large Rajat De\\[-0.3ex]
  \normalsize Stony Brook University,\\[-0.3ex]
  \normalsize Stony Brook, NY, USA\\[-0.3ex]
  \normalsize \texttt{rde@cs.stonybrook.edu}
  \and
  \large Dominik Kempa\thanks{Partially funded by the
  NSF CAREER Award 2337891.}\\[-0.3ex]
  \normalsize Stony Brook University,\\[-0.3ex]
  \normalsize Stony Brook, NY, USA\\[-0.3ex]
  \normalsize \texttt{kempa@cs.stonybrook.edu}
}

\date{\vspace{-0.5cm}}
\maketitle

\input{abstract}

\input{intro}

\input{overview}

\input{prelim}

\input{reductions-from-bmm}

\input{reductions-from-ov}

\input{appendix}

\bibliographystyle{alphaurl}
\bibliography{paper}

\end{document}

%% file: settings.tex
\newtheorem{theorem}{Theorem}[section]
\newtheorem{corollary}[theorem]{Corollary}
\newtheorem{lemma}[theorem]{Lemma}
\newtheorem{observation}[theorem]{Observation}
\newtheorem{proposition}[theorem]{Proposition}
\theoremstyle{definition}
\newtheorem{definition}[theorem]{Definition}
\newtheorem{conjecture}[theorem]{Conjecture}
\theoremstyle{remark}
\newtheorem{remark}[theorem]{Remark}
\newtheorem{example}[theorem]{Example}

\crefname{algocf}{Algorithm}{Algorithms}
\Crefname{algocf}{Algorithm}{Algorithms}

\definecolor{mypink}{RGB}{199, 21, 133}
\hypersetup{
  colorlinks,
  urlcolor = blue,
  linkcolor = mypink,
  citecolor = blue
}

\setul{2pt}{0.6pt}

\setlist[enumerate]{nosep, topsep=1ex}
\setlist[itemize]{nosep, topsep=1ex}
\setlist[description]{nosep,topsep=1ex}

\allowdisplaybreaks

\apptocmd{\sloppy}{\hbadness 10000\relax}{}{}

\let\oldabstract\abstract
\let\oldendabstract\endabstract
\makeatletter
\renewenvironment{abstract}
{%
  {\list{}{\addtolength{\leftmargin}{2.6em}%
    \listparindent 1.5em%
     \itemindent    \listparindent%
     \rightmargin   \leftmargin%
     \parsep        \z@ \@plus\p@}%
     \item\relax}%
  {\endlist}%
\oldabstract}
{\oldendabstract}
\makeatother

\SetCommentSty{textit}
\SetKwInOut{Input}{Input}
\SetKwInOut{Output}{Output}
\SetArgSty{textnormal} %
\SetKw{KwTo}{to}

\makeatletter
\patchcmd\algocf@Vline{\vrule}{\vrule \kern-0.4pt}{}{}
\patchcmd\algocf@Vsline{\vrule}{\vrule \kern-0.4pt}{}{}
\makeatother

%% file: macros.tex
\newcommand{\bigO}{\mathcal{O}}
\newcommand{\dd}{\mathinner{.\,.}}
\newcommand{\probname}[1]{\text{\sc #1}}

\newcommand{\Z}{\mathbb{Z}}
\newcommand{\Zn}{\Z_{\geq 0}}
\newcommand{\Zp}{\Z_{\geq 1}}

\newcommand{\zero}{{\tt 0}}
\newcommand{\one}{{\tt 1}}
\newcommand{\BinaryAlphabet}{\{\zero,\one\}}
\newcommand{\emptystring}{\varepsilon}
\newcommand{\dol}{\text{\normalfont $\texttt{\$}$}}
\newcommand{\hash}{\text{\normalfont $\texttt{\#}$}}
\newcommand{\prefixes}[1]{\mathrm{prefixes}(#1)}
\newcommand{\strword}[2]{\mathrm{stretch}_{#1}(#2)}

\newcommand{\Occurs}[4]{\mathrm{occurs}_{#1}(#2,#3,#4)}
\newcommand{\Rank}[3]{\mathrm{rank}_{#1}(#2,#3)}
\newcommand{\TwoSidedRank}[4]{\mathrm{rank}_{#1}(#2,#3,#4)}
\newcommand{\RangeDistinctCount}[3]{\mathrm{distinct}_{#1}(#2, #3)}
\newcommand{\RangeModeFreq}[3]{\mathrm{mode\mbox{-}freq}_{#1}(#2,#3)}

\newcommand{\RowOnes}[2]{\mathrm{RowOnes}_{#1,#2}}

\newcommand{\OnePos}[1]{\mathrm{OnePos}_{#1}}
\newcommand{\DotProduct}[2]{\left\langle #1, #2 \right\rangle}

\newcommand{\OneVectors}[2]{\mathrm{OneVectors}_{#1,#2}}
\newcommand{\ZeroVectors}[2]{\mathrm{ZeroVectors}_{#1,#2}}

\newcommand{\Rhs}[2]{\mathrm{rhs}_{#1}(#2)}
\newcommand{\Exp}[2]{\mathrm{exp}_{#1}(#2)}
\newcommand{\Lang}[1]{L(#1)}

\newcommand{\BigLZSize}[1]{z_{78}(#1)}

%% file: abstract.tex
\begin{abstract}
  Compressed indexing is a recent trend in the design of data
  structures that aims to support fundamental string queries in space
  proportional to the size of the data in compressed form. One of the
  most popular compression frameworks in this field
  is \emph{grammar compression}. A length-$n$ string
  $T \in \Sigma^{n}$ (where $\Sigma$ is any finite set of size up to
  $|\Sigma| = |T|^{\bigO(1)}$) represented using a context-free grammar of size
  $|G|$ can be augmented to support random access queries (given any
  $i \in [1 \dd n]$, return $T[i]$) in $\bigO(|G| \log^{\bigO(1)} n)$
  space and $\bigO(\log^{\bigO(1)} n)$ time. Numerous other queries,
  including pattern matching, longest common extension,
  lexicographical predecessor/successor, Burrows--Wheeler
  Transform, suffix array, and even suffix tree queries, can also be
  supported within the same bounds.

  Despite this progress, one fundamental class of queries has remained
  elusive: frequency-related queries, such as reporting the number of
  occurrences of a symbol $c \in \Sigma$ in a substring $T(b \dd e]$
  (the so-called \emph{rank} query), or simply checking whether $c$
  occurs in $T(b \dd e]$ (the \emph{symbol occurrence} query). To
  date, no fully general structure achieving $\bigO(|G| \log^{\bigO(1)} n)$ space
  and $\bigO(\log^{\bigO(1)} n)$ query time is known.
  In this work, we establish new conditional lower bounds for frequency-related problems:
  \begin{itemize}

  \item We prove that answering rank and symbol occurrence queries on
    grammar-compressed texts in polylogarithmic time using a
    $\bigO(|G| \log^{\bigO(1)} n)$-space structure that is constructible from
    the input grammar in $\bigO(|G| \log^{\bigO(1)} n)$ time would imply an
    $\bigO(n^2 \log^{\bigO(1)} n)$-time algorithm for Boolean Matrix Multiplication
    (BMM), where the best known algorithms achieve
    $\bigO(n^{2.371339})$ time. Our result is achieved using a more
    general lower bound for efficiently answering a batch of
    rank and symbol occurrence~queries.

  \item We generalize the above result, showing that even
    \emph{LZ78-compressed strings} cannot support efficient rank
    queries. Since LZ78 is provably weaker than grammar compression,
    this yields a stronger result: rank and symbol occurrence queries
    remain hard for a wider class of compressors. We further show that
    achieving even \emph{additive} approximations of rank queries
    would imply faster BMM algorithms.

  \item After establishing hardness of rank and symbol occurrence
    queries, we consider a broader class of frequency-related queries
    and show that, under the popular Orthogonal Vectors (OV)
    conjecture, other problems, including range distinct counting
    and range mode frequency queries, also cannot be efficiently
    supported in compressed space.
  \end{itemize}

  In summary, we develop new techniques for reasoning about
  computation over compressed data, and establish tight connections
  between compressed indexing and long-standing problems in
  fine-grained complexity.
  This sheds new light on compressed indexing
  by isolating a new class of frequency-related queries whose
  complexity hinges on known hard problems.
\end{abstract}

%% file: intro.tex
\section{Introduction}\label{sec:intro}

Text indexing is a classical problem that asks to preprocess a given
length-$n$ sequence (text, string) $T \in \Sigma^{n}$ over an alphabet
$\Sigma$, so that we can efficiently answer various queries on $T$. To
date, numerous indexes using $\bigO(n)$ space are known, supporting a
wide range of queries with query times typically ranging from
$\bigO(1)$ to $\bigO(\log^{\bigO(1)} n)$. These include classical
queries such as suffix arrays/trees~\cite{sa,st,klaap2001}, longest
common extension (LCE)~\cite{st,sss}, pattern
matching~\cite{m2003,BilleGS17,KikiBBGGW14,MunroNN20b},
rank/select~\cite{wavelet,BelazzouguiN14}, lexicographical
predecessor/successor~\cite{csa}, and many
others~\cite{gusfield,Navarro14,ohl2013,AlgorithmsOnStrings,MBCT2023}.

While classical indexes remain fundamental in many applications and
are still frequently used in practice, the need to index massive, highly
repetitive sequences arising from projects such as the 100{,}000
Genomes Project~\cite{100k} or the ongoing 1+ Million Genomes
Initiative~\cite{mg} has led to the development of compressed variants.
A compressed index for a length-$n$ text $T$ is a data structure of size
close to $\bigO(C(T))$ (where $C(T)$ is the output size of some lossless
compression algorithm, or a measure of the repetitiveness of $T$) that
supports efficient queries (typically in $\bigO(\log^{\bigO(1)} n)$ time)
on the original uncompressed text $T$.

The design of a compressed index primarily depends on the underlying
compression representation. One particularly popular framework used
in text indexing is \emph{grammar compression}. In this method, a text
$T \in \Sigma^{n}$ is represented as a context-free grammar (CFG)
whose language consists only of the text $T$. One reason grammar
compression has become a popular framework is its strong theoretical
guarantees: the smallest grammar can be efficiently approximated within a
$\bigO(\log n)$
factor~\cite{Charikar05,Rytter03,Jez16}, and grammar sizes are closely
related to Lempel--Ziv compression~\cite{Rytter03} and the run-length
compressed Burrows--Wheeler transform~\cite{resolution}. More generally,
grammar compression belongs to a broader
family of repetitiveness measures used in compressed indexing. Besides
the size of a smallest grammar, this family includes the LZ77
size~\cite{LZ77}, the run-length BWT size $r(T)$~\cite{BWT}, the size of the
smallest string attractor~\cite{attractors}, and the substring
complexity $\delta(T)$~\cite{delta}, to name a few.
Across a series of works~\cite{Rytter03,Charikar05,GNPlatin18,attractors,resolution,KempaS22,delta},
it has been shown that for these and other
standard measures, the worst-case gaps are only $\bigO(\log^{\bigO(1)} n)$ factors.
Thus, if one ignores polylogarithmic factors, using
smallest grammar size, LZ77 size, run-length BWT size, smallest attractor size,
or $\delta(T)$ as the space benchmark leads to the same notion of
compressed space.

State-of-the-art compressed indexes in this repetitiveness-based setting
support the majority of central string processing queries, including:
\begin{itemize}
\item random
  access~\cite{BLRSRW15,balancing,blocktree,attractors,%
  delta,KempaS22,KempaK23}\footnote{Achieving efficient random access queries in
  $\bigO(n)$ space is trivial: it suffices to store the text
  $T \in \Sigma^{n}$ in plain form. In the compressed setting, however,
  even random access becomes a non-trivial query.};
\item longest common extension
  (LCE)~\cite{tomohiro-lce,dynstr,NishimotoMFCS,KempaS22,KempaK23};
\item pattern matching~\cite{ClaudeN11,ClaudeN12a,ClaudeNP21,%
  GagieGKNP14,GagieGKNP12,ChristiansenEKN21,KociumakaNO22,delta};
\item suffix/LCP array, suffix tree, or Burrows--Wheeler
  transform (BWT) queries~\cite{rindex,KempaK23}.
\end{itemize}
Using the polylogarithmic relations above, their space usage can be
bounded with respect to any grammar $G$ representing a text $T \in
\Sigma^{n}$ (where $\Sigma$ is an alphabet of size up to polynomial in
$|T|$, i.e., $|\Sigma| = |T|^{\bigO(1)}$) by
$\bigO(|G| \cdot \log^{\bigO(1)} n)$, with query times ranging from
$\bigO(\log \log n)$ to $\bigO(\log^{\bigO(1)} n)$. We refer
to surveys of Navarro~\cite{NavarroMeasures,NavarroIndexes} for further details.

Within this common compressed-space regime, the lower-bound picture
began with proving
that random access requires
$\Omega(\tfrac{\log n}{\log \log n})$ time in
$\bigO(\delta(T) \log^{\bigO(1)} n)$ space~\cite{VerbinY13}. More recent work in~\cite{KempaK26}
extends this understanding to most of the
above non-frequency queries and establishes a clean dichotomy. In the
cell-probe model with word size $\Theta(\log n)$, any such index
requires $\Omega(\tfrac{\log n}{\log \log n})$ time for random access,
LCP-array, suffix array, inverse suffix array, and LCE queries, whereas
BWT, PLCP, LF, inverse LF, and lexicographical predecessor/successor
queries require $\Omega(\log \log n)$ time. These bounds match
known upper bounds and already hold over a binary alphabet. Thus,
this work yields two optimal query-time classes,
$\Theta(\tfrac{\log n}{\log \log n})$ and $\Theta(\log \log n)$, for
much of the classical compressed-indexing toolkit.

Despite this progress, one fundamental class of queries has remained
elusive: frequency-related queries. These include reporting the number
of occurrences of a symbol $c \in \Sigma$ in a substring $T(b \dd e]$
(the so-called \emph{rank} query), or simply checking whether $c$
occurs in $T(b \dd e]$ (the \emph{symbol occurrence} query).
Rank queries are among the most widely used queries in string
processing~\cite{FerraginaM05,wt,Navarro14,navarrobook,GonzalezN09,Prezza19}. In the uncompressed setting, these
queries can be supported easily in $\bigO(\log n)$ time by storing the list of
occurrences of each character, and more efficient solutions
are known in the case of an integer alphabet, i.e., when $\Sigma = [0 \dd
\sigma)$~\cite{wavelet,GolynskiMR06,BelazzouguiN15}.

In the compressed setting, however, the understanding of these queries
is significantly more limited due to their dependence on the alphabet
size. The classical queries (such as random access, pattern matching,
suffix array, LCE, or BWT) can be supported in compressed space
independently of the alphabet size, i.e., even when $\Sigma = \{1,
\ldots, \sigma\}$ and $|\Sigma| = |G|^{\bigO(1)}$ or $|\Sigma| =
|T|^{\bigO(1)}$. Rank and symbol occurrence queries, however, appear
to depend strongly on the alphabet size. In the small-alphabet regime,
upper and lower bounds are well understood:
\begin{itemize}
\item Belazzougui et al.~\cite{BelazzouguiCPT15} describe a data
  structure that, for any SLP $G$ representing a string $T \in
  \Sigma^{n}$ with $\Sigma = \{1, \ldots, \sigma\}$, uses
  $\bigO(|\Sigma||G|)$ words of space and answers rank queries in
  $\bigO(\log n)$ time. They also describe a more general trade-off
  using $\bigO(\tau |\Sigma| |G| \log_{\tau}(n/|G|))$ space and
  $\bigO(\log_{\tau} (n/|G|))$ query time, for $2 \leq \tau \leq
  \log^{\epsilon} n$ and any constant $\epsilon > 0$. For $\tau =
  \log^{\epsilon} n$, this yields a structure using $\bigO(|\Sigma| |G|
  (\log^{\epsilon} n) \cdot \log(n / |G|) / \log \log n)$ words of
  space, answering queries in $\bigO(\log(n/|G|) / \log \log n)$ time.
\item On the other hand, Prezza~\cite{Prezza19} generalized the lower
  bound of Verbin and Yu~\cite{VerbinY13} and demonstrated that any
  data structure using $\bigO(|G| \log^{\bigO(1)} n)$ space cannot
  support rank queries in $o((\log n) / \log \log n)$ time. The same
  paper also shows how to achieve trade-offs similar to those
  in~\cite{BelazzouguiCPT15} for a wide range of compressed
  representations by generalizing them to string
  attractors~\cite{attractors}.
\end{itemize}
Consequently, when the alphabet is small, e.g., when $|\Sigma| =
\bigO(\log^{\bigO(1)} n)$, the above solutions yield structures using
$\bigO(|G| \log^{\bigO(1)} n)$ space that support rank queries in the
optimal time $\bigO(\tfrac{\log n}{\log \log n})$. For large alphabets
(say, when $|\Sigma| = |T|^{\bigO(1)}$), the situation is
different:
\begin{itemize}
\item The trade-off from~\cite{BelazzouguiCPT15} in this
  case yields structures using $\widetilde{\bigO}(|\Sigma| |G|)$
  space, which, for example, when $|\Sigma| = |G|$ corresponds to
  quadratic space $\bigO(|G|^2)$. At the other extreme, a
  naive solution using $\bigO(|G|)$ space answers rank queries in
  $\bigO(|G|)$ time.
\item On the hardness side, the authors of~\cite{BelazzouguiCPT15}
  showed that if we can preprocess a grammar
  of size $g$ with $g'$ nonterminals that generates a string of
  length $N$ in $T(g, g', N)$ time and produce a data structure of
  size $S(g, g', N)$ that answers rank queries on the generated string
  in time $t(g, g', N)$, then, given a DAG with $|V|$ nodes, $|E|$
  edges (possibly with multiedges), $\beta$ sources, and $\sigma$
  sinks, we can, after $\bigO(|E| + T(g, g', N))$-time preprocessing,
  produce a data structure of size $\bigO(|E| + S(g, g', N))$ that
  counts the number of distinct paths from any node of the DAG to one
  of the $\sigma$ sinks in time $\bigO(t(g, g', N))$, where $N$ is the
  number of distinct paths that connect the $\beta$ sources to the
  $\sigma$ sinks.
\end{itemize}

In other words, when the alphabet is large, e.g., when $|\Sigma| =
|T|^{\bigO(1)}$, it is currently not known whether rank queries can
be supported in $\bigO(|G| \log^{\bigO(1)} n)$ space and
$\bigO(\log^{\bigO(1)} n)$ query time. Although~\cite{BelazzouguiCPT15}
sheds some light on this hardness by connecting the problem to the DAG
path-counting problem, prior to this work, no precise quantitative lower bounds had been
developed beyond this general reduction, and large-alphabet rank and
symbol occurrence queries remain a central unresolved challenge in
compressed indexing. Furthermore, the known hardness evidence~\cite{BelazzouguiCPT15}
applies only to the relatively powerful rank queries, despite the fact
that no indexes are known even for the much simpler symbol occurrence
queries.

The \emph{large-alphabet} case for rank and symbol occurrence queries
has also recently been shown to be important for 2D string indexing.
In~\cite{indexing2d}, it is proved that if, for a 2D SLP $G_{M}$
representing a 2D string (array, matrix, image) $M \in
\BinaryAlphabet^{r \times c}$, there exists a data structure of size
$\bigO(|G_{M}| \log^{\bigO(1)} n)$ (where $n = \max(r,c)$) that
answers any of the basic 2D queries about subrectangles or subsquares
(including sum, equality, longest common extension, or all-zero
queries), then for any (1D) SLP $G_{T}$ representing a (1D) string $T
\in \Sigma^{*}$, where $|\Sigma| = |T|^{\bigO(1)}$, there exists a
structure of size $\bigO(|G_{T}| \log^{\bigO(1)} |T|)$ that answers
symbol occurrence queries in $\bigO(\log^{\bigO(1)} |T|)$ time. A
similar reduction is proved for rank queries. In other words,
a notion of hardness for rank or symbol occurrence queries
on 1D compressed strings over \emph{(polynomially) large alphabets}
would imply hardness for 2D compressed indexing of 2D strings
over a \emph{binary alphabet}. Given the fundamental role of rank and
symbol occurrence queries in many algorithms~\cite{navarrobook,Navarro14,GonzalezN09,Prezza19,FerraginaM05},
accentuated further by the recent reductions in~\cite{indexing2d}, we thus ask:

\begin{center}
	\emph{Can frequency-related queries (such as rank and symbol occurrence queries)\\
		on large-alphabet strings be efficiently supported in compressed space?}
\end{center}

\paragraph{Our Results}

We present a series of reductions showing that fully general support for
fundamental frequency-related queries over large alphabets (including rank
and symbol occurrence queries, as well as the related problems of range
distinct counting and range mode frequency) would either break
long-standing barriers in computational complexity or require substantially
new approaches.

More specifically, we first prove that efficient support for rank and
symbol occurrence queries would improve the state-of-the-art algorithms for
Boolean matrix multiplication.\footnote{Given any
  $A, B \in \BinaryAlphabet^{n \times n}$,
  Boolean matrix multiplication computes a matrix $C = AB$,
  where $C[i,j] = \bigvee_{k=1}^{n} A[i,k] \wedge B[k,j]$
  holds for every $i,j \in [1 \dd n]$.} Given
any matrices $A, B \in \BinaryAlphabet^{n \times n}$, the currently
best algorithm for this task runs in $\bigO(n^{2.371339})$ time~\cite{fastmm}.
Although we are not aware of any substantial barriers ruling out the existence
of a faster algorithm, and an $\bigO(n^2 \log^{\bigO(1)} n)$-time algorithm for
this problem may exist, our result nevertheless shows that obtaining fast
rank or symbol occurrence queries over grammars would have consequences well beyond compressed
indexing. In this sense, our work is similar in spirit to the conditional
lower bounds for text indexing with mismatches and differences by
Cohen{-}Addad et al.~\cite{Cohen-AddadFS19}.
Specifically, we prove the following theorem.

\begin{restatable}{theorem}{thhardnessofslgsymbolocc}\label{th:hardness-of-slg-symbol-occ}
  If there exists an algorithm that,
  given any SLG $G = (V, \Sigma, R, S)$ generating a string $T \in \Sigma^{N}$
  (where $\Sigma = \{1, \ldots, \sigma\}$ and
  $\sigma = \Omega(N^{1/3})$),
  answers any batch of $m$ symbol occurrence queries
  (\cref{def:symbol-occ}) on $T$
  in $\bigO((|G| + m) \log^{\bigO(1)} N)$ total time,
  then the Boolean matrix product of any two $n \times n$
  Boolean matrices can be computed in $\bigO(n^2 \log^{\bigO(1)} n)$ time.
\end{restatable}

This immediately implies that unless we can multiply Boolean matrices in
$\bigO(n^2 \log^{\bigO(1)} n)$ time, there is no compressed index for symbol
occurrence queries on grammar-compressed text
that is simultaneously small, fast to query, and quickly constructible.

\begin{corollary}
	If there exists a data structure that, given any SLG $G$
  representing a string $T \in \Sigma^{N}$
  (where $\Sigma = \{1, \ldots, \sigma\}$ and $\sigma = \Omega(N^{1/3})$),
  answers symbol occurrence queries
  (\cref{def:symbol-occ}) on $T$
  in $\bigO(\log^{\bigO(1)} N)$ time,
  and takes $\bigO(|G| \log^{\bigO(1)} N)$ time to construct,
  then the Boolean matrix product of any two $n \times n$
  Boolean matrices can be computed in $\bigO(n^2 \log^{\bigO(1)} n)$ time.
\end{corollary}

Since rank queries allow answering symbol occurrence queries, the
above results also hold for rank queries. We state them for
symbol occurrence queries, as this establishes the hardness of these
easier queries (i.e., yields a stronger result). To our knowledge, these
are the first hardness results for symbol occurrence queries, establishing
a surprisingly strong barrier in indexing for these extremely basic
frequency-related queries.

It is worth separating the above result from the well-understood rank
queries used inside BWT-based indexes. The \emph{Burrows--Wheeler
transform (BWT)}~\cite{BWT} is a permutation of the text that plays a
central role in data compression and text indexing~\cite{FerraginaM05,rindex}:
the FM-index of Ferragina and Manzini~\cite{FerraginaM05} relies on
rank over the BWT stored in plain form, whereas the $r$-index of Gagie et
al.~\cite{rindex} relies on rank over its run-length-compressed form,
using $\bigO(r(T))$ or $\bigO(r(T) \log n)$ space, where $r(T)$ is the
number of runs in the BWT of $T$. In both cases, the relevant primitive
is rank on the BWT sequence itself, either uncompressed or only
run-length-compressed, and this setting is well understood from the
upper and lower bound perspectives~\cite{SepulvedaKKP18}. The
surprising point is that $r(T)$ and the smallest grammar size $g^*(T)$
are known to be within $\log^{\bigO(1)} n$ factors of each other in the
worst case~\cite{resolution,GNPlatin18}: thus, rank over the
run-length-compressed BWT is understood, while rank over the original
grammar-compressed text remains challenging.

The hardness is not confined to grammar compression: as explained next,
the above conditional lower bounds for symbol occurrence queries hold even
for significantly weaker compression methods.

\paragraph{Generalization to LZ78}

LZ78~\cite{LZ78} is a classical compression method that, unlike other
compression schemes such as LZ77~\cite{LZ77} or grammar compression,
admits significantly faster algorithms and queries on the underlying text.
For example, the complexity of random access queries on
LZ78-compressed texts (allowing $\bigO(\log^{\bigO(1)} n)$ overhead in space)
is $\Theta(\log \log n)$~\cite{DuttaLRR13,grammarboosting} time, whereas, as noted above, for LZ77,
the optimal query time for random access is
$\Theta(\tfrac{\log n}{\log \log n})$~\cite{BelazzouguiCPT15,blocktree,balancing,VerbinY13}.
This decrease in query time comes at the price of reduced compression ratio:
while LZ77 and grammar compression are capable of exponential
compression, LZ78 cannot compress a length-$n$ string below
$\Omega(\sqrt{n})$ bits.
This motivates us to ask whether rank and symbol occurrence queries can
also be answered more efficiently on LZ78-compressed texts. We answer this
question negatively: we show that the above reduction
from Boolean matrix multiplication holds even on LZ78-compressed text.

\begin{restatable}{theorem}{thhardnessoflzbigsymbolocc}\label{th:hardness-of-lz78-symbol-occ}
  If there exists an algorithm that, given the LZ78 representation
  (\cref{def:lz78-representation}) of a string $T \in \Sigma^{N}$
  (where $\Sigma = \{1, \ldots, \sigma\}$ and $\sigma = \Omega(N^{1/3})$),
  answers any batch of $m$ symbol occurrence queries (\cref{def:symbol-occ})
  on $T$ in $\bigO((\BigLZSize{T} + m) \log^{\bigO(1)} N)$ total time,
  then the Boolean matrix product of any two $n \times n$
  Boolean matrices can be computed in $\bigO(n^2 \log^{\bigO(1)} n)$ time.
\end{restatable}

As in the grammar-compressed case, this immediately implies that unless
we can multiply any two Boolean matrices in $\bigO(n^2 \log^{\bigO(1)} n)$
time, there is no compressed index for symbol occurrence queries on
LZ78-compressed text that is simultaneously small, fast to query, and
quickly constructible. The same implication also holds for rank queries.

\begin{corollary}
	If there exists a data structure that, given the LZ78 representation
  (\cref{def:lz78-representation}) of a string $T \in \Sigma^{N}$
  (where $\Sigma = \{1, \ldots, \sigma\}$ and $\sigma = \Omega(N^{1/3})$),
  answers symbol occurrence queries (\cref{def:symbol-occ})
  on $T$ in $\bigO(\log^{\bigO(1)} N)$ time,
  and takes $\bigO(\BigLZSize{T} \log^{\bigO(1)} N)$ time to construct,
  then the Boolean matrix product of any two $n \times n$
  Boolean matrices can be computed in $\bigO(n^2 \log^{\bigO(1)} n)$ time.
\end{corollary}

\paragraph{Approximate Rank Queries}

The above results raise a natural question of whether
\emph{approximating} rank queries is easier than computing rank values
exactly. Since all of the above results
hold even for symbol occurrence queries (which distinguish whether
$\TwoSidedRank{T}{b}{e}{c} = 0$ or $\TwoSidedRank{T}{b}{e}{c} \geq
1$), we immediately obtain the hardness of multiplicative
approximation (since it would distinguish between the two cases). This
leaves open the possibility of an \ul{additive} approximation of
$\TwoSidedRank{T}{b}{e}{c}$. We show that even additive approximation
is hard.

\begin{restatable}{theorem}{thhardnessofrankapproximation}\label{th:hardness-of-rank-approximation}
  Let $\mu \in (0, 1)$ be a constant.
  If there exists an algorithm that,
  given any SLG $G = (V, \Sigma, R, S)$ generating a string $T \in \Sigma^{N}$
  (where $\Sigma = \{1, \ldots, \sigma\}$ and $\sigma = \Omega(N^{\mu/3})$),
  computes an $\lfloor N^{1-\mu} \rfloor$-additive approximation of any batch of $m$ two-sided rank queries
  (see \cref{def:rank,def:additive-approximation,rm:rank})
  in $\bigO((|G| + m) \log^{\bigO(1)} N)$ total time,
  then the Boolean matrix product of any two $n \times n$
  Boolean matrices can be computed in $\bigO(n^2 \log^{\bigO(1)} n)$ time.
\end{restatable}

\begin{corollary}
  Let $\mu \in (0, 1)$ be a constant.
  If there exists a data structure that, given any SLG $G$
  representing a string $T \in \Sigma^{N}$
  (where $\Sigma = \{1, \ldots, \sigma\}$ and $\sigma = \Omega(N^{\mu/3})$),
  computes an $\lfloor N^{1-\mu} \rfloor$-additive approximation of a given two-sided rank query
  (\cref{def:rank,def:additive-approximation,rm:rank}) in $\bigO(\log^{\bigO(1)} N)$ time,
  and takes $\bigO(|G| \log^{\bigO(1)} N)$ time to construct,
  then the Boolean matrix product of any two $n \times n$
  Boolean matrices can be computed in $\bigO(n^2 \log^{\bigO(1)} n)$ time.
\end{corollary}

\paragraph{Hardness of Other Frequency-Related Queries}

After establishing the hardness of the most basic frequency-related
queries, we turn our attention to other related queries, namely, the
range distinct counting and range mode frequency queries. Consider a
length-$n$ string $T \in \Sigma^{n}$. The \emph{range
	distinct counting} query, given any $b, e \in [0 \dd n]$,
returns $\RangeDistinctCount{T}{b}{e} := |\{T[i] : i \in (b \dd e]\}|$, i.e., the number of
distinct elements in the block $T(b \dd e]$; see
\cref{def:range-distinct-count}. Similarly to the queries considered above,
range distinct counting queries can be answered efficiently in the uncompressed
setting (in~\cite{KaplanRSV07}, the authors describe an algorithm that
achieves $\bigO(n \log n)$ preprocessing time and $\bigO(\log n)$
query time).

We prove that, assuming the popular \emph{Orthogonal Vectors
Conjecture} (\cref{con:ov}), answering a range distinct counting query
on a grammar-compressed string essentially requires
inspecting the entire grammar. As before, we obtain this result as a corollary of the following stronger
batch lower bound:

\begin{restatable}{theorem}{thhardnessofrangedistinctcount}\label{th:hardness-of-range-distinct-count}
  Assuming the Orthogonal Vectors Conjecture (\cref{con:ov}), there is
  no algorithm that, given any SLG $G = (V, \Sigma, R, S)$ representing
  a string $T \in \Sigma^{N}$ (where $\Sigma = \{1, \dots, \sigma\}$ and $\sigma = \Omega((N/\log N)^{1/2})$),
  answers any batch of $m = \Omega(|G|/\log N)$ range distinct count queries (\cref{def:range-distinct-count})
  in $\bigO(m|G|^{1-\epsilon} \log^{\bigO(1)} N)$ time, for any constant $\epsilon > 0$.
\end{restatable}

\begin{corollary}
  Assuming the Orthogonal Vectors Conjecture (\cref{con:ov}), there is
  no data structure that, given any SLG $G$
  representing a string $T \in \Sigma^{N}$
  (where $\Sigma = \{1, \ldots, \sigma\}$ and $\sigma = \Omega((N/\log N)^{1/2})$),
  answers range distinct counting queries
  (\cref{def:range-distinct-count}) on $T$
  in $\bigO(|G|^{1-\epsilon} \log^{\bigO(1)} N)$ time,
  and takes $\bigO(|G|^{2-\epsilon} \log^{\bigO(1)} N)$ time to construct,
  for any $\epsilon > 0$.
\end{corollary}

We complement this hardness result with essentially a matching upper
bound, showing how to answer a batch of $m$ range distinct counting queries in
$\bigO(m|G| \log n)$ time
(see \cref{th:range-distinct-count-on-slg-in-near-linear-time}).

We conclude our set of results by presenting an analogous hardness
argument for \emph{range mode frequency queries}. Given any $b, e \in [0 \dd n]$
satisfying $b < e$, the range mode frequency query returns the frequency of the
most common element in $T(b \dd e]$; see \cref{def:range-mode-freq}.

\begin{restatable}{theorem}{thhardnessofrangemodefreq}\label{th:hardness-of-range-mode-freq}
  Assuming the Orthogonal Vectors Conjecture (\cref{con:ov}), there is
  no algorithm that, given any SLG $G = (V, \Sigma, R, S)$ representing
  a string $T \in \Sigma^{N}$ (where $\Sigma = \{1, \dots, \sigma\}$ and $\sigma = \Omega((N/\log N)^{1/2})$),
  answers any batch of $m = \Omega(|G|/\log N)$ range mode frequency queries (\cref{def:range-mode-freq})
  in $\bigO(m|G|^{1-\epsilon} \log^{\bigO(1)} N)$ time, for any constant $\epsilon > 0$.
\end{restatable}

\begin{corollary}
  Assuming the Orthogonal Vectors Conjecture (\cref{con:ov}), there is
  no data structure that, given any SLG $G$
  representing a string $T \in \Sigma^{N}$
  (where $\Sigma = \{1, \ldots, \sigma\}$ and $\sigma = \Omega((N/\log N)^{1/2})$),
  answers range mode frequency queries
  (\cref{def:range-mode-freq}) on string $T$
  in $\bigO(|G|^{1-\epsilon} \log^{\bigO(1)} N)$ time,
  and takes $\bigO(|G|^{2-\epsilon} \log^{\bigO(1)} N)$ time to construct,
  for any $\epsilon > 0$.
\end{corollary}

\paragraph{Implications of our Hardness Results for Other Range Queries on Grammar-Compressed Strings}

Our hardness results for symbol occurrence queries immediately imply the
hardness of other popular fundamental queries (of which symbol occurrence is
just a special case), such as position-restricted
pattern matching introduced by M{\"a}kinen and Navarro in~\cite{MakinenN06}.
These hardness results hold even on LZ78-compressed strings.

\begin{corollary}
  If there exists a data structure that, given any SLG $G$
  representing a string $T \in \Sigma^{N}$
  (where $\Sigma = \{1, \ldots, \sigma\}$ and $\sigma = \Omega(N^{1/3})$),
  answers position-restricted pattern matching queries
  (that, given any pattern $P \in \Sigma^{M}$ and any pair $b, e \in [0 \dd N-M+1]$,
  checks whether there exists $i \in (b \dd e]$
  satisfying $T[i \dd i + M) = P$)\footnote{We obtain symbol occurrence
  queries as a special case of position-restricted pattern matching queries
  simply by setting $M = 1$; see \cref{def:symbol-occ}.} on $T$
  in $\bigO(M^{\bigO(1)} \log^{\bigO(1)} N)$ time,
  and takes $\bigO(|G| \log^{\bigO(1)} N)$ time to construct,
  then the Boolean matrix product of any two $n \times n$
  Boolean matrices can be computed in $\bigO(n^2 \log^{\bigO(1)} n)$ time.
\end{corollary}

\begin{corollary}
  If there exists a data structure that, given the LZ78 representation
  (\cref{def:lz78-representation}) of a string $T \in \Sigma^{N}$
  (where $\Sigma = \{1, \ldots, \sigma\}$ and $\sigma = \Omega(N^{1/3})$),
  answers position-restricted pattern matching queries
  (that, given any pattern $P \in \Sigma^{M}$ and any pair $b, e \in [0 \dd N-M+1]$,
  checks whether there exists $i \in (b \dd e]$
  satisfying $T[i \dd i + M) = P$) on $T$
  in $\bigO(M^{\bigO(1)} \log^{\bigO(1)} N)$ time,
  and takes $\bigO(\BigLZSize{T} \log^{\bigO(1)} N)$ time to construct,
  then the Boolean matrix product of any two $n \times n$
  Boolean matrices can be computed in $\bigO(n^2 \log^{\bigO(1)} n)$ time.
\end{corollary}

\paragraph{Organization of the Paper}

First, in \cref{sec:overview}, we give an overview of our hardness
reductions. In \cref{sec:prelim}, we formally introduce the notation
and all definitions used in the paper. Next, in
\cref{sec:reductions-from-bmm}, we present the hardness reductions
from the problem of Boolean matrix multiplication (BMM) (specifically,
in \cref{sec:hardness-of-symbol-occ-on-slg}, we prove the hardness of
symbol occurrence queries on grammars, in
\cref{sec:hardness-of-symbol-occ-on-lz78}, we prove similar results
for LZ78, and finally in \cref{sec:hardness-of-approx-rank-queries} we
prove the hardness of approximating rank queries). In
\cref{sec:reductions-from-ov}, we then present our hardness reductions
based on the Orthogonal Vectors Conjecture (specifically, in
\cref{sec:hardness-of-range-distinct-count-on-slg}, we prove the
hardness of range distinct counting queries and in
\cref{sec:hardness-of-range-mode-freq-queries} we show the hardness of
range mode frequency queries).

%% file: overview.tex
\section{Technical Overview}\label{sec:overview}

\paragraph{Reducing Boolean Matrix Multiplication to Symbol Occurrence Queries}

Our main idea is to construct a large but compressible string that
lets us compute a single entry of the product of two matrices using a
\emph{single} symbol occurrence query (see
\cref{def:symbol-occ}). Given two Boolean matrices $A,B \in
\BinaryAlphabet^{n \times n}$, we start by defining the following
objects (see \cref{fig:slg-bmm-reduction} for an example containing
every object defined below):
\begin{enumerate}

\item For each $i \in [1 \dd n]$, we define $\RowOnes{A}{i}$ as the
  string containing (in increasing order) the indices of all columns
  $j$ such that $A[i,j] = \one$
  (see \cref{def:row-ones}).

\item For each $i \in [1 \dd n]$, the string $\RowOnes{B}{i}$ is
  defined analogously, i.e., $\RowOnes{B}{i}$ contains (in increasing
  order) the indices of all columns $j$ such that $B[i,j] = \one$.

\item For each $i \in [1 \dd n]$, we define the string $T_{A,B,i}$ to
  be the concatenation of the strings $\RowOnes{B}{x}$ over all
  indices $x$ appearing in $\RowOnes{A}{i}$, i.e.,
  \[
    T_{A,B,i} := \bigodot_{t=1,\dots,k} \RowOnes{B}{R[t]},
  \]
  where $R = \RowOnes{A}{i}$ and $k = |R|$
  (see \cref{def:mat-mul-answer-string}).

\item Lastly, we define $T_{A,B} := \bigodot_{i=1,\dots,n} T_{A,B,i}$.
\end{enumerate}

We show that there is a direct correspondence between symbols
appearing in $T_{A,B,i}$ and the positions of $\one$-entries in row
$i$ of the Boolean matrix product $AB$. Formally, for every $i,j \in
[1 \dd n]$, the symbol $j$ appears in the string $T_{A,B,i}$ if and
only if $(AB)[i,j] = \one$ holds
(see \cref{lm:sym-occ-and-mat-mul-equivalence}).
We illustrate this correspondence in \cref{fig:slg-bmm-reduction}; for example,
$(AB)[1,2] = (AB)[1,3] = \one$ corresponds to the symbols $2$ and $3$
occurring in $T_{A,B,1}$, and $(AB)[2,1] = \one$ corresponds to $1$
occurring in $T_{A,B,2}$.

After concatenating the strings $T_{A,B,i}$ into $T_{A,B}$, we can
compute the entire product $AB$ using $n^2$ symbol occurrence queries,
one for each pair $(i,j)$, by querying the substring corresponding to
$T_{A,B,i}$. We also construct an SLG of size $\bigO(n^2)$ that
generates $T_{A,B}$
(see \cref{pr:mat-mul-answer-string-slg-construction}).
Thus, any algorithm
that answers a batch of symbol occurrence queries (and consequently
also rank queries; see \cref{def:rank}) on SLG-compressed texts in
amortized polylogarithmic time per query yields a near-quadratic-time
algorithm for Boolean matrix multiplication (see
\cref{th:hardness-of-slg-symbol-occ}).

\begin{figure}
  \centering
  \begin{tikzpicture}
  \newcommand{\topOffset}{2.5}
  \newcommand{\rightOffset}{6}

  \draw (0,\topOffset+1) node[] {A};
  \node [matrix,nodes={inner sep=2pt},left delimiter = {[}, right delimiter = {]}] (matA) at (0,\topOffset){
    \node{1}; & \node{0}; & \node{1}; \\
    \node{0}; & \node{1}; & \node{0}; \\
    \node{1}; & \node{1}; & \node{0}; \\
  };

  \draw (2.5,\topOffset+1) node[] {B};
  \node [matrix,nodes={inner sep=2pt},left delimiter = {[}, right delimiter = {]}] (matB) at (2.5,\topOffset){
    \node{0}; & \node{1}; & \node{0}; \\
    \node{1}; & \node{0}; & \node{0}; \\
    \node{0}; & \node{1}; & \node{1}; \\
  };

  \node [matrix,nodes={inner sep=2pt},anchor=west] (RowOnesA) at (\rightOffset-2,\topOffset){
    \node{$\RowOnes{A}{1} =\ $13};\\
    \node{$\RowOnes{A}{2} =\ $2};\\
    \node{$\RowOnes{A}{3} =\ $12};\\
  };
  
  \node [matrix,nodes={inner sep=2pt},anchor=west] (RowOnesB) at (\rightOffset+2,\topOffset){
    \node{$\RowOnes{B}{1} =\ $2};\\
    \node{$\RowOnes{B}{2} =\ $1};\\
    \node{$\RowOnes{B}{3} =\ $23};\\
  };

  \draw (1.25,1) node[] {AB};
  \node [matrix,nodes={inner sep=2pt},left delimiter = {[}, right delimiter = {]}] (matAB) at (1.25,0){
    \node{0}; & \node{1}; & \node{1}; \\
    \node{1}; & \node{0}; & \node{0}; \\
    \node{1}; & \node{1}; & \node{0}; \\
  };

  \node [matrix,nodes={inner sep=2pt},anchor=west] (TABis) at (\rightOffset-2,0){
    \node{$T_{A,B,1} =\ $223};\\
    \node{$T_{A,B,2} =\ $1};\\
    \node{$T_{A,B,3} =\ $21};\\
  };

  \draw (\rightOffset + 3.5,0)node {$T_{A,B} =\ $ 223121};
  \end{tikzpicture}
  \caption{\label{fig:slg-bmm-reduction}Example showing matrices
    $A$ and $B$ along with strings $\RowOnes{A}{i}$, $\RowOnes{B}{i}$
    (\cref{def:row-ones}),
    $T_{A,B,i}$ for $i \in [1 \dd 3]$, and $T_{A,B}$
    (\cref{def:mat-mul-answer-string}).}
\end{figure}

\paragraph{Hardness for Symbol Occurrence Queries on LZ78-Compressed Text}

We use a new variant of the \emph{grammar boosting} technique
of~\cite{grammarboosting} to transform the structured BMM instance
$(A,B)$ underlying the answer string $T_{A,B}$ into a new string
$X_{A,B}$
(see \cref{def:mat-mul-lz78-answer-string}).
The key idea is to add prefix gadgets that force LZ78 to create, for every
$i \in [1 \dd n]$, the phrase $\dol_i \cdot \RowOnes{B}{i}$ together
with all of its prefixes. Once these phrases have been created, each
later block of the form $\dol_p \cdot \RowOnes{B}{p} \cdot \hash_i$ is
parsed as a single additional LZ78 phrase, because it extends an
already existing phrase by one fresh delimiter. This yields an
explicit description of the LZ78 representation of $X_{A,B}$ and
allows us to compute it in $\bigO(n^2)$ time from $A$ and $B$
(see \cref{lm:lz78-representation}, \cref{pr:lz78-parsing-of-mat-mul-answer-string}, and \cref{alg:lz78-parsing-of-answer-string}).
Moreover, for every $i \in [1 \dd n]$, the block $X_{A,B,i}$
preserves the occurrences of all symbols $j \in [1 \dd n]$ from the
corresponding block $T_{A,B,i}$. Hence, the symbol occurrence queries
used in the BMM reduction can be simulated on $X_{A,B}$, and the same
idea yields the corresponding hardness for rank queries (see
\cref{th:hardness-of-lz78-symbol-occ}).

Applying this technique to the example string $T_{A,B}$ from
\cref{fig:slg-bmm-reduction} results in the following transformed
string (with $\odot$ added for clarity):
\[
  \dol_1 \dol_1 2 \dol_2 \dol_2 1 \dol_3 \dol_3 2 \dol_3 23 \odot
  \dol_1 2 \hash_1 \dol_3 2 3 \hash_1 \odot
  \dol_2 1 \hash_2 \odot
  \dol_1 2 \hash_3 \dol_2 1 \hash_3
\]
The LZ78 parsing of the above string is (with parentheses denoting
each phrase):
\[
  (\dol_1) (\dol_1 2) (\dol_2) (\dol_2 1) (\dol_3) (\dol_3 2) (\dol_3 23)
  (\dol_1 2 \hash_1) (\dol_3 2 3 \hash_1) (\dol_2 1 \hash_2) (\dol_1 2 \hash_3) (\dol_2 1 \hash_3)
\]
Lastly, we highlight the portion of this transformed string that
corresponds to the symbols of $T_{A,B}$ from
\cref{fig:slg-bmm-reduction}. All symbols appearing between the
highlighted ones are auxiliary delimiters. This illustrates how symbol
occurrence queries on the transformed string can simulate symbol
occurrence queries on $T_{A,B}$ (and the same simulation applies to
rank queries).
\[
  \dol_1 \dol_1 2 \dol_2 \dol_2 1 \dol_3 \dol_3 2 \dol_3 23\ 
  \dol_1 {\color{red}{2}} \hash_1\ 
  \dol_3 {\color{red}{2 3}} \hash_1\ 
  \dol_2 {\color{red}{1}} \hash_2\ 
  \dol_1 {\color{red}{2}} \hash_3\ 
  \dol_2 {\color{red}{1}} \hash_3
\]

\paragraph{Hardness of Approximate Rank Queries}

The hardness results for symbol occurrence queries
imply that any multiplicative approximation for two-sided rank queries
is also hard, since it would
distinguish between the cases $\TwoSidedRank{T}{b}{e}{c}=0$ and
$\TwoSidedRank{T}{b}{e}{c} \geq 1$. We therefore consider additive
approximations. We take any string $T \in \Sigma^{N}$ and replace every
symbol with $k$ copies of itself (call this transformed string
$\strword{k}{T}$;
see \cref{def:strword}).
After this transformation,
a substring of $T$ contains a symbol $c \in \Sigma$ if and only if the
corresponding substring in $\strword{k}{T}$ contains at least $k$
copies of $c$
(see \cref{lm:stretch}).
If the original string $T$ can
be generated by a small grammar, then $\strword{k}{T}$ can also be
generated by a small grammar, and such a grammar can be constructed
efficiently
(see \cref{pr:slg-stretching}).
Letting $T$ be the string defined in \cref{def:mat-mul-answer-string}, we
obtain hardness
results for additive approximations of two-sided rank queries over
grammar-compressed texts (see
\cref{th:hardness-of-rank-approximation}).

\paragraph{Hardness for Range Distinct Count and Range Mode Frequency Queries}

Let $A = (a_1,\dots,a_n)$ be a sequence of $n \geq 1$ binary vectors
of dimension $d \geq 1$. To show hardness of range distinct count queries (see
\cref{def:range-distinct-count}), we construct a string over alphabet
$\{1,\dots,n\}$ consisting of $n$ \emph{blocks} such that, for each
$i$, the symbols missing from block $i$ are exactly the indices of
vectors orthogonal to $a_i$. We start by defining the following
objects (see \cref{fig:ovreduction1} for an example showing all the
objects defined below):

\begin{enumerate}

\item For each $j \in [1 \dd d]$, we define $\OneVectors{A}{j}$ to be
  the string containing (in increasing order) the indices $i$ of all
  vectors satisfying $a_i[j] = \one$
  (see \cref{def:one-vectors}).

\item For any vector $x \in \BinaryAlphabet^d$, we define $\OnePos{x}$
  to be the string containing (in increasing order) the indices of all
  coordinates $j$ satisfying $x[j] = \one$
  (see \cref{def:one-pos}).

\item For every $i \in [1 \dd n]$, we define the string $U_{A,i}$ as
  follows:
  \[
    U_{A,i} := \bigodot_{j=1,\dots,k} \OneVectors{A}{R[j]},
  \]
  where $R = \OnePos{a_i}$ and $k = |R|$.

\item Lastly, we let $U_A := \bigodot_{i=1,\dots,n}U_{A,i}$
  (see \cref{def:range-distinct-count-answer-string}).
\end{enumerate}

We prove that for every $i,j \in [1 \dd n]$, the symbol $j$ does
\emph{not} occur in $U_{A,i}$ if and only if $\DotProduct{a_i}{a_j} = 0$
(see \cref{lm:reduce-ov-to-range-distinct-count}).
Therefore, $\RangeDistinctCount{U_{A,i}}{0}{|U_{A,i}|} < n$ holds if and only if
$a_i$ is orthogonal to some vector in $A$
(see \cref{lm:ov-and-range-distinct-count-equivalence}).
Hence, after
concatenating all blocks into $U_A$, we can use $n$ queries
$\RangeDistinctCount{U_A}{\cdot}{\cdot}$
(\cref{def:range-distinct-count}) to solve the orthogonal vectors
problem on $A$.

We demonstrate this reduction in \cref{fig:ovreduction1}. Here, the
strings $\OnePos{a_i}$ contain the indices of all coordinates $j$ such
that $a_i[j] = \one$ for $i \in [1 \dd 5]$, and the strings
$\OneVectors{A}{j}$ contain the indices of all vectors $a_i$ such that
$a_i[j] = \one$ for $j \in [1 \dd 4]$. The strings $U_{A,i}$ are
defined as above
(\cref{def:range-distinct-count-answer-string}).
We observe that $U_{A,1}$ contains all elements from $\{1,\dots,5\}$,
which matches the fact that $a_1$ is not orthogonal to any
vector in $A$. The element $4$ is missing from $U_{A,2}$ (and symmetrically,
$2$ is missing from $U_{A,4}$). This corresponds to the fact that
$\DotProduct{a_2}{a_4} = 0$. Similarly, the element $5$ is missing
from $U_{A,3}$ (and $3$ is missing from $U_{A,5}$), which corresponds
to $\DotProduct{a_3}{a_5} = 0$. Thus, in this example, the only pairs
of orthogonal vectors are $(a_2,a_4)$ and $(a_3,a_5)$.

\begin{figure}[t]
  \centering
  \begin{tikzpicture}
    \newcommand{\topOffset}{3}
    \newcommand{\rightOffset}{6}

    \draw node at (2,\topOffset){$a_1 = [\one,\zero,\zero,\one]$};
    \draw node at (2,\topOffset-0.5){$a_2 = [\one,\one,\zero,\zero]$};
    \draw node at (2,\topOffset-1){$a_3 = [\zero,\one,\zero,\one]$};
    \draw node at (2,\topOffset-1.5){$a_4 = [\zero,\zero,\one,\one]$};
    \draw node at (2,\topOffset-2){$a_5 = [\one,\zero,\one,\zero]$};

    \draw node at (\rightOffset,\topOffset){$\OnePos{a_1}$ = 14};
    \draw node at (\rightOffset,\topOffset-0.5){$\OnePos{a_2}$ = 12};
    \draw node at (\rightOffset,\topOffset-1){$\OnePos{a_3}$ = 24};
    \draw node at (\rightOffset,\topOffset-1.5){$\OnePos{a_4}$ = 34};
    \draw node at (\rightOffset,\topOffset-2){$\OnePos{a_5}$ = 13};

    \draw node[anchor=west] at (\rightOffset+2,\topOffset){$\OneVectors{A}{1}$ = 125};
    \draw node[anchor=west] at (\rightOffset+2,\topOffset-0.5){$\OneVectors{A}{2}$ = 23};
    \draw node[anchor=west] at (\rightOffset+2,\topOffset-1){$\OneVectors{A}{3}$ = 45};
    \draw node[anchor=west] at (\rightOffset+2,\topOffset-1.5){$\OneVectors{A}{4}$ = 134};

    \draw node[anchor=west] at (1,0){$U_{A,1}$ = 125134};
    \draw node[anchor=west] at (1,-0.5){$U_{A,2}$ = 12523};
    \draw node[anchor=west] at (1,-1){$U_{A,3}$ = 23134};
    \draw node[anchor=west] at (1,-1.5){$U_{A,4}$ = 45134};
    \draw node[anchor=west] at (1,-2){$U_{A,5}$ = 12545};

    \draw node at (\rightOffset+1.5,-1){$U_{A}$ = 125134 12523 23134 45134 12545};
    \draw[decorate,decoration={brace}] +(\rightOffset+2.5,-1.25) -- node[below]{Missing 5} ++(\rightOffset+1.6,-1.25);
    \draw node at (\rightOffset+1.5,-1){$U_{A}$ = 125134 12523 23134 45134 12545};
    \draw[decorate,decoration={brace}] +(\rightOffset+2.65,-0.75) -- node[above]{Missing 2} ++(\rightOffset+3.55,-0.75);
  \end{tikzpicture}
  \caption{\label{fig:ovreduction1}Example showing our reduction
    from Orthogonal Vectors to range distinct count queries.}
\end{figure}

The string $U_A$ can be generated by a grammar $G$ of size $\bigO(nd)$
that can also be constructed in $\bigO(nd)$ time
(see \cref{pr:range-distinct-count-answer-string-slg-construction}).
Thus, if we can answer a batch of $n$ queries
$\RangeDistinctCount{U_A}{\cdot}{\cdot}$ with average time
$\bigO(|G|^{1-\epsilon} \log^{\bigO(1)} N)$ per query
(where $N$ is the length of the input text), then we can solve the Orthogonal
Vectors problem in sub-quadratic time (see
\cref{th:hardness-of-range-distinct-count}).

The hardness for range mode frequency queries (see
\cref{def:range-mode-freq}) follows an analogous structure. We replace
each string $\OneVectors{A}{j}$ with its complement: we let
$\ZeroVectors{A}{j}$ be the string containing (in increasing order)
the indices of all vectors $a_i$ such that $a_i[j] = \zero$
(see \cref{def:zero-vectors}).
We then define the corresponding blocks $W_{A,i}$
(see \cref{def:range-mode-freq-answer-string}).
We show that symbol $j$ appears in $W_{A,i}$ less than
$|\OnePos{a_i}|$ times if and only if
$\DotProduct{a_i}{a_j} \neq 0$. Thus, the maximum frequency in
$W_{A,i}$ is at least $|\OnePos{a_i}|$ if and only if $a_i$ is
orthogonal to some vector in $A$.
Proceeding as above, we
obtain hardness of answering a batch of range mode frequency
queries over substrings of grammar-compressed texts (see
\cref{th:hardness-of-range-mode-freq}).

%% file: prelim.tex
\section{Preliminaries}\label{sec:prelim}

\subsection{Basic Definitions}\label{sec:prelim-basic}

\subsubsection{Strings}\label{sec:prelim-strings}

A \emph{string} is a finite sequence of characters drawn from a given
\emph{alphabet} $\Sigma$. The length of a string $S$ is denoted by
$|S|$. For $i \in [1 \dd |S|]$,%
\footnote{For $i,j \in \mathbb{Z}$, we define
$[i \dd j] = \{k \in \mathbb{Z} : i \leq k \leq j\}$,
$[i \dd j) = \{k \in \mathbb{Z} : i \leq k < j\}$, and
$(i \dd j] = \{k \in \mathbb{Z} : i < k \leq j\}$.}
the $i$th leftmost character of $S$ is denoted $S[i]$. A \emph{substring} of $S$
is any string of the form $S[i \dd j) = S[i]S[i{+}1]\cdots S[j{-}1]$
for some $1 \leq i \leq j \leq |S|{+}1$.
Substrings of the forms $S[1 \dd j)$ and $S[i \dd |S|{+}1)$ are called
\emph{prefixes} and \emph{suffixes}, respectively.
The \emph{concatenation} of two strings $S_1$ and $S_2$,
namely the string $S_1[1]\cdots S_1[|S_1|]S_2[1]\cdots S_2[|S_2|]$,
is denoted by $S_1S_2$ or $S_1 \cdot S_2$.
For $k \in \mathbb{Z}_{\ge 0}$, we define
$S^k = \bigodot_{i=1}^k S$ as the concatenation of $k$ copies of $S$;
by convention, $S^0 = \emptystring$ denotes the \emph{empty string}.

\subsubsection{Matrices and Vectors}\label{sec:prelim-matrices-and-vectors}

A matrix $A \in \BinaryAlphabet^{n \times m}$ is a two-dimensional
array with $n$ rows and $m$ columns. The entry at the intersection of
row $i$ and column $j$ is denoted as $A[i,j]$ for every $(i,j) \in [1
\dd n] \times [1 \dd m]$. All matrix products used in this paper are
over the Boolean semiring, i.e., for $A,B \in \BinaryAlphabet^{n
\times n}$, we define $C=AB$ as $C[i,j] = \bigvee_{k=1}^{n} A[i,k]
\wedge B[k,j]$ for each $i,j \in [1 \dd n]$.

A binary vector $v$ of dimension $d$ is an element of
$\BinaryAlphabet^{d}$. The $i$th coordinate of $v$ is
denoted as $v[i]$. For two vectors $u,v$ of dimension $d$, we denote
their dot product by $\DotProduct{u}{v} = \sum_{i=1}^d u[i] \cdot v[i]$.
We say two vectors $u$ and $v$ are \emph{orthogonal} if $\DotProduct{u}{v} = 0$.

\subsection{Frequency-Based Queries}\label{sec:prelim-frequency-based-queries}

\subsubsection{Symbol Occurrence and Rank Queries}\label{sec:prelim-symbol-occ-and-rank}

\begin{definition}[Symbol occurrence]\label{def:symbol-occ}
  Let $T \in \Sigma^{n}$.
  For every $b, e \in [0 \dd n]$ and every $c \in \Sigma$, we define
  \vspace{1ex}
  \[
    \Occurs{T}{b}{e}{c} := 
      \begin{cases}
        1 & \text{if there exists } j \in (b \dd e] \text{ such that }T[j] = c,\\
        0 & \text{otherwise}.
      \end{cases}
  \]
\end{definition}
\vspace{1ex}

\begin{example}\label{ex:symbol-occ}
  For $T = \texttt{acaaba}$ it holds $\Occurs{T}{1}{4}{\texttt{a}} = 1$
  and $\Occurs{T}{1}{4}{\texttt{b}} = 0$.
\end{example}

\begin{definition}[Rank]\label{def:rank}
  Let $T \in \Sigma^{n}$.
  For every $b, e \in [0 \dd n]$ and every $c \in \Sigma$, we define
  \[
    \TwoSidedRank{T}{b}{e}{c} := |\{i \in (b \dd e] : T[i] = c\}|.
  \]
  For every $j \in [0 \dd n]$ and every $c \in \Sigma$, we also define
  \vspace{-0.5ex}
  \[
    \Rank{T}{j}{c} := \TwoSidedRank{T}{0}{j}{c} = |\{i \in (0 \dd j] : T[i] = c\}|.
  \]
\end{definition}

\begin{example}\label{ex:rank}
  For $T = \texttt{acaaba}$, it holds
  $\Rank{T}{5}{\texttt{a}} = 3$ and
  $\TwoSidedRank{T}{2}{5}{\texttt{a}} = 2$.
\end{example}

\begin{remark}\label{rm:rank}
  To distinguish between the two types of rank queries in \cref{def:rank},
  by a \emph{two-sided rank query}, we mean the computation of $\TwoSidedRank{T}{b}{e}{c}$,
  and by \emph{one-sided rank query}, we refer to the computation of $\Rank{T}{j}{c}$.
  Note that, for every $T \in \Sigma^{n}$, $c \in \Sigma$, and $b, e \in [0 \dd n]$
  such that $b \leq e$, it holds
  $\TwoSidedRank{T}{b}{e}{c} = \Rank{T}{e}{c} - \Rank{T}{b}{c}$.
\end{remark}

\begin{definition}[Additive rank approximation]\label{def:additive-approximation}
  Let $T \in \Sigma^{n}$, $c \in \Sigma$, and $b, e \in [0 \dd n]$.
  Let $k \in \Zp$. We say that $x \in \Z$ is a \emph{$k$-additive approximation}
  of $\TwoSidedRank{T}{b}{e}{c}$ if it holds
  \[
    \TwoSidedRank{T}{b}{e}{c} - k < x < \TwoSidedRank{T}{b}{e}{c} + k.
  \]
\end{definition}

\subsubsection{Range Distinct Count Queries}\label{sec:prelim-range-count-distinct}

\begin{definition}[Range distinct count]\label{def:range-distinct-count}
  Let $T \in \Sigma^{n}$. For any $b, e \in [0 \dd n]$, we define
  \[
    \RangeDistinctCount{T}{b}{e} := |\{T[i] : i \in (b \dd e]\}|.
  \]
\end{definition}

\begin{example}\label{ex:count-distinct}
  For $T = \texttt{acaaba}$, it holds
  $\RangeDistinctCount{T}{0}{4} = 2$ and
  $\RangeDistinctCount{T}{1}{5} = 3$.
\end{example}

\subsubsection{Range Mode Frequency Queries}\label{sec:prelim-range-mode}

\begin{definition}[Range mode frequency]\label{def:range-mode-freq}
  Let $T \in \Sigma^{n}$. For every $b, e \in [0 \dd n]$,
  we define $\RangeModeFreq{T}{b}{e}$ as the frequency of the most common
  element in $T(b \dd e]$, i.e.,
  \[
    \RangeModeFreq{T}{b}{e} := \max_{c \in \Sigma} \TwoSidedRank{T}{b}{e}{c}
  \]
  (see \cref{def:rank}).
\end{definition}

\begin{example}\label{ex:range-mode-count}
  For $T = \texttt{acaaba}$ it holds
  $\RangeModeFreq{T}{0}{4} = 3$ and
  $\RangeModeFreq{T}{3}{6} = 2$.
\end{example}

\subsection{Compressed Representations}\label{sec:prelim-compressed-representations}

\subsubsection{Grammar Compression}\label{sec:prelim-grammar-compression}

A \emph{context-free grammar (CFG)} is a tuple $G = (V, \Sigma,
R, S)$ such that $V \cap \Sigma = \emptyset$ and
\begin{itemize}
  \item $V$ is a finite nonempty set of \emph{nonterminals} or \emph{variables},
  \item $\Sigma$ is a finite nonempty set of \emph{terminal} symbols,
  \item $R \subseteq V \times (V \cup \Sigma)^*$ is a set of
    \emph{productions} or \emph{rules}, and
  \item $S \in V$ is the special \emph{starting nonterminal}.
\end{itemize}

We say that $u \in (V \cup \Sigma)^{*}$ \emph{derives} $v$, and write
$u \Rightarrow^* v$, if $v$ can be obtained from $u$ by repeatedly
replacing nonterminals according to the rule set $R$. We then denote
$\Lang{G} := \{w \in \Sigma^* \mid S \Rightarrow^* w\}$.

By a \emph{straight-line grammar (SLG)} we mean a CFG $G =
(V,\Sigma,R,S)$ such that:
\begin{enumerate}
\item there exists an ordering $(N_1, \dots, N_{|V|})$ of all elements
  in $V$ such that, for every $(N,\gamma) \in R$, letting $i \in [1
    \dd |V|]$ be such that $N = N_i$, it holds $\gamma \in (\{N_{i+1},
  \dots, N_{|V|}\} \cup \Sigma)^{*}$, and
\item for every nonterminal $N \in V$, there exists exactly one
  $\gamma \in (V \cup \Sigma)^{*}$ such that $(N,\gamma) \in R$.
\end{enumerate}
The unique string $\gamma \in (V \cup \Sigma)^{*}$ such that
$(N,\gamma) \in R$ is called the \emph{definition} or \emph{right-hand side}
of nonterminal $N$ and is denoted $\Rhs{G}{N}$. Note that in an SLG, for every $\alpha
\in (V \cup \Sigma)^{*}$, there exists exactly one string $\gamma \in
\Sigma^{*}$ satisfying $\alpha \Rightarrow^{*} \gamma$. Such $\gamma$
is called the \emph{expansion} of $\alpha$ and is denoted
$\Exp{G}{\alpha}$. In particular, in an SLG, we have $|\Lang{G}| = 1$.
We define the size of an SLG as $|G| = \sum_{N \in V}\max(|\Rhs{G}{N}|, 1)$.

An SLG in which, for every $N \in V$, it holds $\Rhs{G}{N} = XY$,
where $X,Y \in V$, or $\Rhs{G}{N} = a$, where $a \in \Sigma$, is
called a \emph{straight-line program (SLP)}.

\begin{observation}\label{ob:slg-to-slp}
  Every SLG $G = (V, \Sigma, R, S)$ satisfying $\Lang{G} \neq \{\emptystring\}$
  can be transformed in $\bigO(|G|)$ time into an
  SLP $G' = (V', \Sigma, R', S')$
  satisfying $|G'| = \Theta(|G|)$ and $\Lang{G'} = \Lang{G}$.
\end{observation}

\begin{observation}\label{ob:slp-min-size}
  If $G$ is an SLP such that $\Lang{G} = \{T\}$, where $|T| = n$, then
  $|G| = \Omega(\log n)$.
\end{observation}

\begin{lemma}[{\cite{balancing}}]\label{lm:slp-balancing}
  Given any SLP $G = (V, \Sigma, R, S)$ representing a string $T \in \Sigma^{n}$,
  we can in $\bigO(|G|)$ time construct an SLP
  $G' = (V', \Sigma, R', S')$ of height $\bigO(\log n)$ that
  represents the same string $T$.
\end{lemma}

\subsubsection{Lempel--Ziv (LZ78) Compression}\label{sec:prelim-lz78}

\begin{definition}[Lempel--Ziv (LZ78) factorization~\cite{LZ78}]\label{def:lz78}
  The \emph{LZ78 factorization} of a string $T$ is a factorization
  $T = f_1 f_2 \cdots f_k$ defined such that, letting $f_0 = \emptystring$,
  for every $i \in [1 \dd k]$, after $f_1\cdots f_{i-1}$ has been parsed,
  $f_i$ is the longest prefix of the remaining suffix $f_i f_{i+1} \cdots f_k$
  such that there exists $j \in [0 \dd i)$ satisfying $f_i = f_j \cdot c$ for
  some symbol $c$. The elements $f_i$ of the factorization (where $i \geq 1$) are called
  \emph{phrases}. We denote the number of phrases in the LZ78 factorization
  of $T$ by $\BigLZSize{T}$ (i.e., $k = \BigLZSize{T}$).
\end{definition}

\begin{definition}[Lempel--Ziv (LZ78) representation]\label{def:lz78-representation}
  Let $T = f_1 f_2 \cdots f_k$ be the LZ78 factorization of $T$ (\cref{def:lz78}), and
  let $f_0 = \emptystring$.
  The \emph{LZ78 representation} of $T$ is a sequence of pairs
  $((p_1,c_1), \dots, (p_k,c_k))$ such that, for every $i \in [1 \dd k]$,
  it holds $p_i \in [0 \dd i)$ and $f_{i} = f_{p_i} \cdot c_i$.
\end{definition}

\begin{remark}\label{rm:lz78-representation-uniqueness}
  Observe that if $T = f_1 f_2 \cdots f_k$ is the LZ78 factorization of $T$, then
  all phrases $f_1, \dots, f_{k-1}$ are distinct. Consequently, the LZ78 representation (\cref{def:lz78-representation})
  of $T$ is unique.
\end{remark}

\begin{remark}\label{rm:lz78-encoding-bit-size}
  Note that if $T \in [0 \dd \sigma)^{n}$ for some $\sigma \in \Z_{\geq 1}$,
  then the LZ78 representation of $T$ (\cref{def:lz78-representation})
  encodes $T$ using $\bigO(\BigLZSize{T} \cdot (\log \BigLZSize{T} + \log \sigma)) =
  \bigO(\BigLZSize{T} \cdot (\log n + \log \sigma))$
  bits of space.
\end{remark}

\begin{example}\label{ex:lz78}
  The LZ78 factorization (\cref{def:lz78}) of the string $T = \texttt{010212020}$ is
  $T =
    \texttt{0} \cdot
    \texttt{1} \cdot
    \texttt{02} \cdot
    \texttt{12} \cdot
    \texttt{020}
  $
  with $\BigLZSize{T} = 5$ phrases, and the LZ78 representation (\cref{def:lz78-representation})
  of $T$ is:
  $(
      (0,\texttt{0}),
      (0,\texttt{1}),
      (1,\texttt{2}),
      (2,\texttt{2}),
      (3,\texttt{0})
  )$.
\end{example}

\subsection{Hardness Assumptions}\label{sec:prelim-hardness-assumptions}

\subsubsection{Boolean Matrix Multiplication (BMM)}

The first problem we use as the basis of our hardness arguments is
the Boolean Matrix Multiplication problem.

\begin{framed}
  \noindent
  \probname{Boolean Matrix Multiplication (BMM)}\par\smallskip
  \begin{description}[
    topsep=0pt,
    partopsep=0pt,
    parsep=0pt,
    itemsep=0pt,
    leftmargin=!,
    labelwidth=\widthof{\bfseries Output:}
  ]
  \item[Input:] Two matrices $A,B \in \BinaryAlphabet^{n \times n}$.
  \item[Output:] The matrix $C = AB$, where
    $C[i,j] = \bigvee_{k=1}^{n} A[i,k] \wedge B[k,j]$
    for every $i,j \in [1 \dd n]$.
  \end{description}
\end{framed}

The state-of-the-art for the above problem is summarized below, and
currently no algorithm is known that solves this problem in
$\bigO(n^2 \log^{\bigO(1)} n)$ time.

\begin{theorem}[{\cite{fastmm}}]
  Given any two Boolean matrices $A, B \in \BinaryAlphabet^{n \times n}$,
  we can compute their product in $\bigO(n^{2.371339})$ time.
\end{theorem}

\subsubsection{Orthogonal Vectors (OV)}

The second problem we use as a basis for our hardness results is the
Orthogonal Vectors problem, one
of the most widely used tools in fine-grained complexity theory
(see, e.g., \cite{Williams24} and references therein for a recent
discussion).

\begin{framed}
  \noindent
  \probname{Orthogonal Vectors (OV)}\par\smallskip
  \begin{description}[
    topsep=0pt,
    partopsep=0pt,
    parsep=0pt,
    itemsep=0pt,
    leftmargin=!,
    labelwidth=\widthof{\bfseries Output:}
  ]
  \item[Input:] A set of vectors $A \subseteq \BinaryAlphabet^d$ with $|A| = n$.
  \item[Output:] Determine whether there exist $x,y \in A$ such that
    $\DotProduct{x}{y} = \sum_{i=1}^{d} x[i] \cdot y[i] = 0$.
  \end{description}
\end{framed}

\begin{conjecture}[Orthogonal Vectors Conjecture]\label{con:ov}
  For every constant $\epsilon > 0$, there exists a constant $c \geq 1$ such
  that OV cannot be solved in $\bigO(n^{2-\epsilon})$ time on instances
  with $d = c\log n$.
\end{conjecture}

\subsection{Model of Computation}

We use the standard word RAM model of computation~\cite{Hagerup98}
with $w$-bit machine words, where $w = \Theta(\log n)$, and all
standard bitwise and arithmetic operations take $\bigO(1)$ time.
Unless explicitly stated otherwise, we measure space complexity
in machine words.

%% file: reductions-from-bmm.tex
\section{Reductions from Boolean Matrix Multiplication}\label{sec:reductions-from-bmm}

\input{hardness-of-symbol-occ-on-slg}

\input{hardness-of-symbol-occ-on-lz78}

\input{hardness-of-approx-rank}

%% file: hardness-of-symbol-occ-on-slg.tex
\subsection{Hardness of Symbol Occurrence Queries on Grammars}\label{sec:hardness-of-symbol-occ-on-slg}

\begin{definition}[String with positions of all $\one$s in a given row]\label{def:row-ones}
  Let $A \in \BinaryAlphabet^{n \times n}$, where $n \geq 1$.
  For every $i \in [1 \dd n]$, by
  $\RowOnes{A}{i}$, we denote a string
  containing in increasing order the positions of all ones in the $i$th row
  of $A$, i.e., $\RowOnes{A}{i} \in \{1, \dots, n\}^{*}$ is such that,
  letting $k = |\RowOnes{A}{i}|$, it holds
  \begin{itemize}
  \item if $k > 0$, then $\RowOnes{A}{i}[1] < \dots < \RowOnes{A}{i}[k]$, and
  \item $\{\RowOnes{A}{i}[t]\}_{t \in [1 \dd k]} =
    \{j \in [1 \dd n] : A[i,j] = \one\}$.
  \end{itemize}
\end{definition}

\begin{definition}[Grammar-compressible encoding of the product $AB$ as a string]\label{def:mat-mul-answer-string}
  Let $A, B \in \BinaryAlphabet^{n \times n}$, where $n \geq 1$.
  For every $i \in [1 \dd n]$, we denote
  \[
    T_{A,B,i} := \textstyle\bigodot_{t=1,\dots,k} \RowOnes{B}{R[t]},
  \]
  where $R = \RowOnes{A}{i}$ (\cref{def:row-ones}) and $k = |R|$.
  We then let
  \[
    T_{A,B} := \textstyle\bigodot_{i=1,\dots,n} T_{A,B,i}.
  \]
\end{definition}

\begin{observation}\label{ob:mat-mul-answer-string-length}
  For every $A,B \in \BinaryAlphabet^{n \times n}$, the string $T_{A,B}$
  (\cref{def:mat-mul-answer-string}) satisfies $|T_{A,B}| \leq n^3$.
\end{observation}
\begin{proof}
  For all $i \in [1 \dd n]$, it holds $|\RowOnes{A}{i}| \leq n$ and
  $|\RowOnes{B}{i}| \leq n$. Since $T_{A,B,i}$ is a concatenation of at most
  $n$ strings each of which has length at most $n$, we thus have
  $|T_{A,B,i}| \leq n^2$ for all $i \in [1 \dd n]$.
  Consequently, $|T_{A,B}| \leq n^3$.
\end{proof}

\begin{lemma}\label{lm:sym-occ-and-mat-mul-equivalence}
  Let $A,B \in \BinaryAlphabet^{n \times n}$.
  For every $i,j \in [1 \dd n]$, the following conditions are equivalent:
  \begin{enumerate}
  \item\label{lm:sym-occ-and-mat-mul-equivalence-it-1}
    Symbol $j$ occurs in the string $T_{A,B,i}$
    (see \cref{def:mat-mul-answer-string}),
  \item\label{lm:sym-occ-and-mat-mul-equivalence-it-2}
    It holds $(AB)[i,j] = \one$
    (where $AB$ denotes the Boolean matrix multiplication).
  \end{enumerate}
\end{lemma}
\begin{proof}
  (\ref{lm:sym-occ-and-mat-mul-equivalence-it-1}) $\Rightarrow$
  (\ref{lm:sym-occ-and-mat-mul-equivalence-it-2})
  Assume that symbol $j$ occurs in the string $T_{A,B,i}$.
  Denote $R = \RowOnes{A}{i}$ and $k = |R|$.
  By \cref{def:mat-mul-answer-string}, the assumption that $j$ occurs in
  $T_{A,B,i}$ implies that there exists $t \in [1 \dd k]$ such that $j$
  occurs in the string $\RowOnes{B}{R[t]}$. By definition of
  $\RowOnes{B}{R[t]}$, this implies that $B[R[t],j] = \one$.
  On the other hand, by definition of $R = \RowOnes{A}{i}$, we have
  $A[i,R[t]] = \one$. Consequently, $A[i,R[t]] \wedge B[R[t],j] = \one$.
  Hence, $(AB)[i,j] = \one$.

  (\ref{lm:sym-occ-and-mat-mul-equivalence-it-2}) $\Rightarrow$
  (\ref{lm:sym-occ-and-mat-mul-equivalence-it-1})
  Let us now assume that $(AB)[i,j] = \one$. This implies that there exists
  $p \in [1 \dd n]$ such that $A[i,p] = \one$ and $B[p,j] = \one$.
  Denote again $R = \RowOnes{A}{i}$ and $k = |R|$.
  On the one hand, $A[i,p] = \one$ implies that $p$ occurs in $R$,
  i.e., there exists $t \in [1 \dd k]$ such that $p = R[t]$.
  On the other hand, by $B[p,j] = \one$, we obtain that $j$ occurs in
  $\RowOnes{B}{p} = \RowOnes{B}{R[t]}$. 
  We have thus proved that there exists $t \in [1 \dd k]$ such that
  $j$ occurs in $\RowOnes{B}{R[t]}$. Since $\RowOnes{B}{R[t]}$ is a
  substring of $T_{A,B,i}$, we thus obtain that $j$ occurs in $T_{A,B,i}$.
\end{proof}

\begin{proposition}\label{pr:mat-mul-answer-string-slg-construction}
  Given any $A, B \in \BinaryAlphabet^{n \times n}$, we can in
  $\bigO(n^2)$ time construct an SLG $G = (V, \Sigma, R, S)$
  such that $\Sigma = \{1, \dots, n\}$ and
  $\Lang{G} = \{T_{A,B}\}$ (\cref{def:mat-mul-answer-string}).
\end{proposition}
\begin{proof}

  Let $V = \{X_1, \dots, X_n, Y_1, \dots, Y_n, S\}$.
  For every $i \in [1 \dd n]$, we set
  $\Rhs{G}{X_i} = \RowOnes{B}{i} \in \{1, \dots, n\}^{*}$
  (\cref{def:row-ones}). For every $i \in [1 \dd n]$,
  letting $R = \RowOnes{A}{i}$ and $k = |R|$, we set
  $\Rhs{G}{Y_i} = X_{R[1]} \cdot X_{R[2]} \cdots X_{R[k]}$,
  with the right-hand side interpreted as $\emptystring$ if $k=0$.
  Finally, we set $\Rhs{G}{S} = Y_1 \cdot Y_2 \cdots Y_n$.

  To construct $G$, we proceed as follows:
  \begin{enumerate}
  \item First, in $\bigO(n^2)$ time we compute
    the strings $\RowOnes{A}{i}$ and $\RowOnes{B}{i}$ for
    all $i \in [1 \dd n]$.
  \item Given the above strings, the construction of
    $\Rhs{G}{N}$ for all $N \in V$ takes
    $\bigO(n^2)$ time.
  \end{enumerate}
  In total, the construction takes $\bigO(n^2)$ time.

  To show that $\Lang{G} = \{T_{A,B}\}$, it suffices to observe
  that, for every $i \in [1 \dd n]$, it holds
  $\Exp{G}{Y_i} = T_{A,B,i}$ (see \cref{def:mat-mul-answer-string}).
  Thus, $\Lang{G} = \{\Exp{G}{S}\} =
  \{\bigodot_{i=1,\dots,n}\Exp{G}{Y_i}\} = \{T_{A,B}\}$.
\end{proof}

\thhardnessofslgsymbolocc*
\begin{proof}
  Let $A,B \in \BinaryAlphabet^{n \times n}$ be two given Boolean matrices.
  The algorithm to compute the Boolean matrix product $AB$ proceeds as follows:
  \begin{enumerate}

  \item In this step, we compute an array
    $A_{\rm sum}[0 \dd n]$ defined such that $A_{\rm sum}[0] = 0$ and,
    for every $i \in [1 \dd n]$, $A_{\rm sum}[i] = \sum_{j=1}^{i} |T_{A,B,j}|$ (see \cref{def:mat-mul-answer-string}).
    We begin by computing an array $A_{\rm len}[1 \dd n]$ defined by
    $A_{\rm len}[i] = |T_{A,B,i}|$.
    To this end, we first in $\bigO(n^2)$ time compute the strings $\RowOnes{A}{i}$ and $\RowOnes{B}{i}$
    (see \cref{def:row-ones}) for
    all $i \in [1 \dd n]$. Using these strings, we can then compute $|T_{A,B,i}|$ for any $i \in [1 \dd n]$
    in $\bigO(n)$ time (see \cref{def:mat-mul-answer-string}). In total, computing
    $A_{\rm len}[1 \dd n]$ takes $\bigO(n^2)$ time.
    Using $A_{\rm len}$, we then easily obtain $A_{\rm sum}$ in $\bigO(n)$ time.
    In total, construction of $A_{\rm sum}[0 \dd n]$ takes $\bigO(n^2)$ time.
    Observe that by \cref{def:mat-mul-answer-string}, for any $i \in [1 \dd n]$, it holds
    \[
      T_{A,B}(A_{\rm sum}[i-1] \dd A_{\rm sum}[i]] = T_{A,B,i}.
    \]
    If $A_{\rm sum}[n] = 0$, then $T_{A,B,i}$ is empty for every
    $i \in [1 \dd n]$. By \cref{lm:sym-occ-and-mat-mul-equivalence},
    the output matrix $AB$ then consists only of zeros. In this case,
    we return the $n \times n$ zero matrix in $\bigO(n^2)$ time and
    conclude the algorithm. Henceforth, assume that $A_{\rm sum}[n] > 0$.

  \item Using \cref{pr:mat-mul-answer-string-slg-construction},
    in $\bigO(n^2)$ time we construct an SLG
    $G = (V, \Sigma, R, S)$ such that $\Sigma = \{1, \dots, \sigma\}$ (where $\sigma = n$) and
    $\Lang{G} = \{T\}$, where $T = T_{A,B}$ (\cref{def:mat-mul-answer-string}).
    Note that the upper bound on the runtime of \cref{pr:mat-mul-answer-string-slg-construction} implies that $|G| = \bigO(n^2)$.
    Denote $N = |T|$ and recall that $N \leq n^3$ (\cref{ob:mat-mul-answer-string-length}).
    Observe also that we then have $\sigma = n = \Omega(N^{1/3})$.

  \item In this step, we compute the product $AB$. By
    \cref{lm:sym-occ-and-mat-mul-equivalence},
    computing $(AB)[i,j]$ for any given $i,j \in [1 \dd n]$
    can be done using a single symbol occurrence query. More precisely,
    by \cref{lm:sym-occ-and-mat-mul-equivalence} and the above discussion, it holds
    (see \cref{def:symbol-occ})
    \begin{align*}
      (AB)[i,j]
        &= \Occurs{T_{A,B,i}}{0}{|T_{A,B,i}|}{j}\\
        &= \Occurs{T_{A,B}}{A_{\rm sum}[i-1]}{A_{\rm sum}[i]}{j}\\
        &= \Occurs{T}{A_{\rm sum}[i-1]}{A_{\rm sum}[i]}{j}.
    \end{align*}
    Using the array $A_{\rm sum}$, in $\bigO(n^2)$ time we prepare
    the arguments for a batch of $m = n^2$ symbol occurrence queries on $T$.
    Using the algorithm from the claim, we then answer all the queries
    (and hence compute the product $AB$) in
    $
      \bigO((|G| + m) \log^{\bigO(1)} N) =
      \bigO(n^2 \log^{\bigO(1)} (n^3)) =
      \bigO(n^2 \log^{\bigO(1)} n)
    $
    time.
 \end{enumerate}

 In total, the algorithm takes $\bigO(n^2 \log^{\bigO(1)} n)$ time.
\end{proof}

%% file: hardness-of-symbol-occ-on-lz78.tex
\subsection{Hardness of Symbol Occurrence Queries on LZ78}\label{sec:hardness-of-symbol-occ-on-lz78}

\begin{definition}[Prefix concatenation string]\label{def:prefixes}
  For any string $T \in \Sigma^{n}$, we denote
  \[
    \prefixes{T} :=
      \textstyle\bigodot_{i=1,\dots,n} T[1 \dd i] \in \Sigma^{n(n+1)/2}.
  \]
\end{definition}

\begin{lemma}\label{lm:lz78-representation}
  Let $m, k \in \Z_{\geq 1}$.
  Let $(S_i)_{i \in [1 \dd m]}$ be a sequence of strings over alphabet
  $\Sigma = \{1, \ldots, \sigma\}$. Let $\dol_1, \ldots, \dol_m, \hash_1, \ldots, \hash_k$
  be pairwise distinct symbols that do not belong to $\Sigma$. Denote
  $
    \hat{\Sigma} = \Sigma \cup
      \{\dol_1,  \dol_2,  \ldots, \dol_m \} \cup
      \{\hash_1, \hash_2, \ldots, \hash_k\}
  $.
  For every $i \in [1 \dd k]$, let $n_i \in \Z_{\geq 0}$ and
  $(p_{i,j})_{j \in [1 \dd n_i]}$ be a possibly empty sequence of
  integers such that, whenever $n_i > 0$,
  $
    1 \leq p_{i,1} < p_{i,2} < \cdots < p_{i,n_i} \leq m
  $.
  For every $i \in [1 \dd k]$,
  we define the following string (parentheses added for clarity):
  \[
    T_{i} :=
      \textstyle\bigodot_{j=1,\ldots,n_i}
        \Big(\dol_{p_{i,j}} \cdot S_{p_{i,j}} \cdot \hash_{i}\Big)
          \in \hat{\Sigma}^{*}.
  \]
  We also define (see \cref{def:prefixes}):
  \[
    T :=
      \Big(\textstyle\bigodot_{i=1,\ldots,m} \prefixes{\dol_i \cdot S_i}\Big)
      \cdot
      \Big(\bigodot_{i=1,\ldots,k} T_{i}\Big)
        \in \hat{\Sigma}^{*}.
  \]
  Then, the LZ78 factorization (\cref{def:lz78}) of $T$ is (with empty terms omitted):
  \[
    \begin{aligned}
      T ={}\;&
        \begin{array}[t]{@{}l c l c l c l c l c@{}}
          \dol_1 & \odot & \dol_1 S_1[1] & \odot & \dol_1 S_1[1\dd 2] & \odot & \dol_1 S_1[1\dd 3] & \cdots & \dol_1 S_1[1\dd |S_1|] & \odot \\
          \dol_2 & \odot & \dol_2 S_2[1] & \odot & \dol_2 S_2[1\dd 2] & \odot & \dol_2 S_2[1\dd 3] & \cdots & \dol_2 S_2[1\dd |S_2|] & \odot \\
          \vdots &      & \vdots        &      & \vdots              &      & \vdots              &       & \vdots                 &      \\
          \dol_m & \odot & \dol_m S_m[1] & \odot & \dol_m S_m[1\dd 2] & \odot & \dol_m S_m[1\dd 3] & \cdots & \dol_m S_m[1\dd |S_m|] & \odot
        \end{array}
        \\[0.6ex]
        &\begin{array}[t]{@{}l c l c l c l c@{}}
          \dol_{p_{1,1}}   S_{p_{1,1}}   \hash_{1} & \odot &
          \dol_{p_{1,2}}   S_{p_{1,2}}   \hash_{1} & \odot &
          \dol_{p_{1,3}}   S_{p_{1,3}}   \hash_{1} & \cdots &
          \dol_{p_{1,n_1}} S_{p_{1,n_1}} \hash_{1} & \odot \\
          \dol_{p_{2,1}}   S_{p_{2,1}}   \hash_{2} & \odot &
          \dol_{p_{2,2}}   S_{p_{2,2}}   \hash_{2} & \odot &
          \dol_{p_{2,3}}   S_{p_{2,3}}   \hash_{2} & \cdots &
          \dol_{p_{2,n_2}} S_{p_{2,n_2}} \hash_{2} & \odot \\
          \vdots &      & \vdots &      & \vdots &      & \vdots &      \\
          \dol_{p_{k,1}}   S_{p_{k,1}}   \hash_{k} & \odot &
          \dol_{p_{k,2}}   S_{p_{k,2}}   \hash_{k} & \odot &
          \dol_{p_{k,3}}   S_{p_{k,3}}   \hash_{k} & \cdots &
          \dol_{p_{k,n_k}} S_{p_{k,n_k}} \hash_{k}. &
        \end{array}
    \end{aligned}
  \]
  In particular,
  $\BigLZSize{T} = m + \sum_{j=1}^{m} |S_j| + \sum_{i=1}^{k} n_i$.
  Moreover, letting
  $(a_i)_{i \in [0 \dd m]}$ be a sequence defined such that, for any $i \in [0 \dd m]$, $a_i = \sum_{j=1}^{i} |S_j|$,
  in the LZ78 representation (\cref{def:lz78-representation}) of $T$:
  \begin{itemize}
  \item a phrase $\dol_i$, where $i \in [1 \dd m]$, is encoded as $(0,\dol_i)$,
  \item a phrase $\dol_i \cdot S_i[1 \dd j]$, where $i \in [1 \dd m]$ and $j \in [1 \dd |S_i|]$, is encoded as $(a_{i-1}+(i-1)+j, S_i[j])$,
  \item a phrase $\dol_{p_{i,j}} \cdot S_{p_{i,j}} \cdot \hash_i$, where $i \in [1 \dd k]$ and $j \in [1 \dd n_i]$, is encoded as $(a_{p_{i,j}}+p_{i,j},\hash_i)$.
  \end{itemize}
\end{lemma}
\begin{proof}

  For each $i \in [1 \dd m]$, let $X_i := \dol_i \cdot S_i$. We first
  prove by induction on $i$ that, for every $i \in [1 \dd m]$, the
  first $i+a_i$ phrases in the LZ78 factorization of $T$ are exactly
  the phrases $X_1[1], X_1[1 \dd 2], \ldots, X_1[1 \dd |X_1|], \ldots,
  X_i[1],\allowbreak X_i[1 \dd 2], \ldots, X_i[1 \dd |X_i|]$, i.e.,
  $\dol_1, \dol_1S_1[1], \ldots, \dol_1S_1, \ldots, \allowbreak \dol_i, \dol_iS_i[1], \ldots, \dol_iS_i$.
  \begin{itemize}

  \item For the base case $i=1$, note that the prefix of $T$
    corresponding to $\prefixes{X_1}$ is exactly $X_1[1] \cdot
    X_1[1 \dd 2] \cdots X_1[1 \dd |X_1|]$. The first phrase is
    therefore $X_1[1]=\dol_1$, since $\dol_1$ has not occurred
    earlier. After this, the unread suffix again starts with
    $\dol_1$. More generally, whenever the phrases $X_1[1], X_1[1 \dd
    2], \ldots, X_1[1 \dd r]$ have already been produced, the next
    unread position starts with $X_1[1 \dd r+1]$; this string is
    obtained by extending the earlier phrase $X_1[1 \dd r]$ by one
    symbol, and no longer phrase can be chosen, because among all
    earlier phrases the only ones starting with $\dol_1$ are precisely
    $X_1[1], X_1[1 \dd 2], \ldots, X_1[1 \dd r]$. Hence the phrases
    contributed by $\prefixes{X_1}$ are exactly $X_1[1], X_1[1 \dd
    2], \ldots, X_1[1 \dd |X_1|]$.

  \item For the induction step, assume that the claim holds for some
    $i \in [1 \dd m)$, and let $\delta := i+a_i$. Then the first
    $\delta$ phrases are exactly those coming from
    $\prefixes{X_1}, \ldots, \prefixes{X_i}$. The next unread part of
    $T$ begins with $\prefixes{X_{i+1}}$, namely with
    $X_{i+1}[1] \cdot X_{i+1}[1 \dd 2] \cdots X_{i+1}[1 \dd
    |X_{i+1}|]$. Since $\dol_{i+1}$ does not occur in the already
    parsed prefix, the next phrase is
    $X_{i+1}[1]=\dol_{i+1}$. Repeating the same argument as in the
    base case, after the phrases $X_{i+1}[1], X_{i+1}[1 \dd
    2], \ldots, X_{i+1}[1 \dd r]$ have been produced, the next unread
    suffix starts with $X_{i+1}[1 \dd r+1]$; this is obtained by
    extending the previous phrase $X_{i+1}[1 \dd r]$, and it is
    maximal because the only earlier phrases starting with
    $\dol_{i+1}$ are the prefixes already created in this block. Thus
    the phrases contributed by $\prefixes{X_{i+1}}$ are exactly
    $X_{i+1}[1], X_{i+1}[1 \dd 2], \ldots, X_{i+1}[1 \dd
    |X_{i+1}|]$. This proves the induction claim.
  \end{itemize}

  Let $\Delta := m+a_m$. By the above, the first $\Delta$
  phrases in the LZ78 factorization of $T$ are exactly
  $\dol_1, \dol_1S_1[1], \ldots, \dol_1S_1, \ldots, \dol_m, \dol_mS_m[1], \ldots, \dol_mS_m$,
  and the remaining suffix is $\bigodot_{i=1,\ldots,k} T_i$. Consider
  any substring of this suffix of the form $\dol_t \cdot
  S_t \cdot \hash_r$, where $t \in [1 \dd m]$ and $r \in [1 \dd
  k]$. Since $\dol_tS_t$ is already one of the first $\Delta$ phrases,
  the string $\dol_tS_t\hash_r$ can be chosen as the next LZ78
  phrase. It is also maximal: indeed, if a longer phrase were to start
  here, then there would exist an earlier phrase $Q$ that matches a
  longer prefix of the unread suffix and has $\dol_tS_t\hash_r$ as a
  prefix. Since every nonempty LZ78 phrase is obtained from an earlier
  phrase by appending one symbol, repeatedly following source phrases
  from $Q$ shows that every nonempty prefix of $Q$ is also an earlier
  phrase. In particular, the string $\dol_tS_t\hash_r$ itself would
  have to occur earlier as a phrase. This is impossible, because
  $\hash_r$ does not occur in the prefix
  $\bigodot_{i=1,\ldots,m} \prefixes{\dol_iS_i}$, and within
  the suffix $\bigodot_{i=1,\ldots,k} T_i$ each substring
  $\dol_tS_t\hash_r$ appears at most once: any occurrence of
  $\dol_tS_t\hash_r$ must start at an occurrence of $\dol_t$, and
  since each $S_u$ belongs to $\Sigma^*$ while all symbols
  $\dol_1,\ldots,\dol_m,\hash_1,\ldots,\hash_k$ lie outside $\Sigma$,
  the symbol $\dol_t$ can occur in the suffix only as the first symbol
  of one of the blocks $\dol_{p_{u,v}}S_{p_{u,v}}\hash_u$; for fixed
  $r$ this follows from the strict inequalities $p_{r,1} < \cdots <
  p_{r,n_r}$, and for different values of $r$ the last symbol differs
  since $\hash_1,\ldots,\hash_k$ are distinct. Therefore each
  substring $\dol_{p_{i,j}} S_{p_{i,j}} \hash_i$ forms one phrase of
  the LZ78 factorization of $T$, in the stated order. This proves the
  claimed formula for the factorization. In particular, the number of
  phrases is $\Delta + \sum_{i=1}^{k} n_i = m + \sum_{j=1}^{m} |S_j|
  + \sum_{i=1}^{k} n_i$.

  It remains to verify the formulas for the LZ78 representation.
  \begin{itemize}

  \item A phrase $\dol_i$, where $i \in [1 \dd m]$, is encoded as $(0,\dol_i)$
    by definition.

  \item Next, let us consider a phrase $\dol_i \cdot S_i[1 \dd j]$, where
    $i \in [1 \dd m]$ and $j \in [1 \dd |S_i|]$. This is the phrase
    $X_i[1 \dd j+1]$, and it is the $(a_{i-1}+(i-1)+(j+1))$th phrase in
    the factorization. Its source phrase is the immediately preceding
    phrase $X_i[1 \dd j]$, whose index is $a_{i-1}+(i-1)+j$. Therefore
    its encoding is $(a_{i-1}+(i-1)+j, S_i[j])$.

  \item Finally, let us consider a phrase
    $\dol_{p_{i,j}} \cdot S_{p_{i,j}} \cdot \hash_i$, where $i \in
    [1 \dd k]$ and $j \in [1 \dd n_i]$. Its source is the phrase
    $\dol_{p_{i,j}} \cdot S_{p_{i,j}}$, and by the first part of the
    proof that source is the $(a_{p_{i,j}} + p_{i,j})$th phrase. Hence
    the encoding of $\dol_{p_{i,j}} \cdot S_{p_{i,j}} \cdot \hash_i$ is
    $(a_{p_{i,j}} + p_{i,j}, \hash_i)$.
    \qedhere
  \end{itemize}
\end{proof}

\begin{definition}[LZ78-compressible encoding of the product $AB$ as a string]\label{def:mat-mul-lz78-answer-string}
  Let $A, B \in \BinaryAlphabet^{n \times n}$, where $n \geq 1$.
  For every $i \in [1 \dd n]$, we define the string (parentheses added for clarity):
  \[
    X_{A,B,i} := \textstyle\bigodot_{t=1,\dots,k} \Big( (n+R[t]) \cdot \RowOnes{B}{R[t]} \cdot (2n + i) \Big)
                 \in [1 \dd 3n]^{*},
  \]
  where $R = \RowOnes{A}{i}$ (\cref{def:row-ones}) and $k = |R|$.
  We then let (see \cref{def:prefixes})
  \[
    X_{A,B} := \Big( \textstyle\bigodot_{i=1,\ldots,n} \prefixes{(n+i) \cdot \RowOnes{B}{i}} \Big) \cdot
               \Big( \textstyle\bigodot_{i=1,\dots,n} X_{A,B,i}                              \Big)
               \in [1 \dd 3n]^{*}.
  \]
\end{definition}

\begin{observation}\label{ob:lz78-answer-string-length}
  For every $A,B \in \BinaryAlphabet^{n \times n}$, the string $X_{A,B}$
  (\cref{def:mat-mul-lz78-answer-string}) satisfies $|X_{A,B}| \leq 6n^3$.
\end{observation}
\begin{proof}
  For all $i \in [1 \dd n]$, it holds $|\RowOnes{A}{i}| \leq n$ and
  $|\RowOnes{B}{i}| \leq n$. Thus, $X_{A,B,i}$ is a concatenation of at most
  $n$ strings each of which has length at most $n+2$, and hence
  $|X_{A,B,i}| \leq n(n+2)$ for all $i \in [1 \dd n]$.
  On the other hand, by \cref{def:prefixes}, for every $i \in [1 \dd n]$,
  letting $k = |\RowOnes{B}{i}|$, it holds $|\prefixes{(n+i) \cdot \RowOnes{B}{i}}| = (k+1)(k+2)/2 \leq (n+1)(n+2)/2$.
  Consequently, $|X_{A,B}| \leq n(n+1)(n+2)/2 + n^2(n+2) = \tfrac{3}{2}n^3 + \tfrac{7}{2}n^2 + n \leq 6n^3$.
\end{proof}

\begin{figure}
  \begin{algorithm}[H]
    \caption{Computing the LZ78 representation of the string $X_{A,B}$ (\cref{def:mat-mul-lz78-answer-string}).}\label{alg:lz78-parsing-of-answer-string}
    \Input{Boolean matrices $A,B \in \BinaryAlphabet^{n \times n}$.}
    \Output{The LZ78 representation (\cref{def:lz78-representation}) of the string $X_{A,B}$ (\cref{def:mat-mul-lz78-answer-string}).}
    \vspace{1ex}
    $\mathrm{sum}[0] \gets 0$\\
    \For{$i \gets 1$ \KwTo $n$}{
      $\ell \gets 0$\\
      $\mathrm{parse}.\mathrm{append}(0,n+i)$\\
      \For{$j \gets 1$ \KwTo $n$}{
        \If{$B[i,j] = \one$}{
          $\ell \gets \ell + 1$\\
          $\mathrm{parse}.\mathrm{append}(\mathrm{sum}[i-1] + (i-1) + \ell, j)$\\
        }
      }
      $\mathrm{sum}[i] \gets \mathrm{sum}[i-1] + \ell$\\
    }
    \For{$i \gets 1$ \KwTo $n$}{
      \For{$j \gets 1$ \KwTo $n$}{
        \If{$A[i,j] = \one$}{
          $\mathrm{parse}.\mathrm{append}(\mathrm{sum}[j] + j, 2n + i)$\\
        }
      }
    }
    \Return{$\mathrm{parse}$}
  \end{algorithm}
\end{figure}

\begin{proposition}\label{pr:lz78-parsing-of-mat-mul-answer-string}
  Given any $A, B \in \BinaryAlphabet^{n \times n}$, we can compute the
  LZ78 representation (\cref{def:lz78-representation}) of the string
  $X_{A,B}$ (\cref{def:mat-mul-lz78-answer-string}) in $\bigO(n^2)$ time.
\end{proposition}
\begin{proof}
  Observe that letting $S_i = \RowOnes{B}{i}$, $n_i =
  |\RowOnes{A}{i}|$, and $(p_{i,j})_{j=1,\ldots,n_i}$ (where $i \in
  [1 \dd n]$) be the sequence of symbols in $\RowOnes{A}{i}$,
  the string $T$ from \cref{lm:lz78-representation} is equal to the
  string $X_{A,B}$, assuming we map symbols in the sets
  $\{\dol_1, \dots, \dol_n\}$ and $\{\hash_1, \dots, \hash_n\}$ so
  that for every $i \in [1 \dd n]$, $\dol_i$ (resp. $\hash_i$) is
  mapped to $n + i$ (resp.\ $2n + i$).
  By \cref{lm:lz78-representation}, we can thus compute the LZ78
  representation of $X_{A,B}$ as follows:
  \begin{enumerate}

  \item In $\bigO(n^2)$ time, we compute, for $i \in [1 \dd n]$, the
    strings $A_i = \RowOnes{A}{i}$ and $B_i = \RowOnes{B}{i}$.

  \item In $\bigO(n)$ time, we compute an array ${\rm sum}[0 \dd n]$
    defined such that ${\rm sum}[0] = 0$ and, for every $i \in [1 \dd n]$, ${\rm sum}[i] = {\rm sum}[i-1] + |B_i|$.

  \item For $i=1,\ldots,n$, perform the following steps:
    \begin{enumerate}
    \item First, append the pair $(0, n+i)$ to the output LZ78 representation.
    \item For $j=1,\ldots,|B_i|$, append the pair
      $({\rm sum}[i-1]+(i-1)+j,B_i[j])$ to the output LZ78 representation.
    \end{enumerate}
    In total, this takes $\bigO(n^2)$ time.

  \item For $i=1,\ldots,n$, scan the sequence $A_i$
    left-to-right and, for every $j \in [1 \dd |A_i|]$, append the
    pair $({\rm sum}[A_i[j]] + A_i[j], 2n + i)$ to the output LZ78
    representation.  In total, this takes $\bigO(n^2)$
    time.
  \end{enumerate}

  In total, the computation takes $\bigO(n^2)$ time.
  An equivalent optimized implementation, which avoids explicitly storing the strings $A_i$
  and $B_i$, is given in \cref{alg:lz78-parsing-of-answer-string}.
\end{proof}

\thhardnessoflzbigsymbolocc*
\begin{proof}
  Let $A,B \in \BinaryAlphabet^{n \times n}$ be two given Boolean matrices.
  The algorithm to compute the Boolean matrix product $AB$ proceeds as follows:
  \begin{enumerate}

  \item In the first step, we compute the integer $\Delta = \sum_{i=1}^{n} |\prefixes{(n+i) \cdot \RowOnes{B}{i}}|$.
    To this end, we first in $\bigO(n^2)$ time compute the strings $\RowOnes{B}{i}$ for all $i \in [1 \dd n]$.
    Using their lengths, the computation of $\Delta$ takes $\bigO(n)$ time (see \cref{def:prefixes}).

  \item In the second step, we compute an array
    $A_{\rm sum}[0 \dd n]$ defined such that $A_{\rm sum}[0] = 0$ and,
    for every $i \in [1 \dd n]$, $A_{\rm sum}[i] = \sum_{j=1}^{i} |X_{A,B,j}|$ (see \cref{def:mat-mul-lz78-answer-string}).
    We begin by computing an array $A_{\rm len}[1 \dd n]$ defined by
    $A_{\rm len}[i] = |X_{A,B,i}|$.
    To this end, we first in $\bigO(n^2)$ time compute the strings $\RowOnes{A}{i}$ and $\RowOnes{B}{i}$
    (see \cref{def:row-ones}) for
    all $i \in [1 \dd n]$. Using these strings, we can then compute $|X_{A,B,i}|$ for any $i \in [1 \dd n]$
    in $\bigO(n)$ time (see \cref{def:mat-mul-lz78-answer-string}). In total, computing
    $A_{\rm len}[1 \dd n]$ takes $\bigO(n^2)$ time.
    Using $A_{\rm len}$, we then easily obtain $A_{\rm sum}$ in $\bigO(n)$ time.
    In total, construction of $A_{\rm sum}[0 \dd n]$ takes $\bigO(n^2)$ time.
    Observe that by \cref{def:mat-mul-lz78-answer-string}, for any $i \in [1 \dd n]$, it holds
    \[
      X_{A,B}(\Delta + A_{\rm sum}[i-1] \dd \Delta + A_{\rm sum}[i]] = X_{A,B,i}.
    \]

  \item Using \cref{pr:lz78-parsing-of-mat-mul-answer-string},
    in $\bigO(n^2)$ time we compute the LZ78 representation of the string $T = X_{A,B}$ (\cref{def:mat-mul-lz78-answer-string}).
    Note that $T$ is over alphabet $\Sigma = \{1, \dots, \sigma\}$, where $\sigma = 3n$.
    Note also that by \cref{def:mat-mul-lz78-answer-string} and \cref{lm:lz78-representation}, it follows that $\BigLZSize{T} =
    n + \sum_{i=1}^{n} |\RowOnes{B}{i}| + \sum_{i=1}^{n} |\RowOnes{A}{i}| \leq n + 2n^2$.
    Denote $N = |T|$ and recall that $N \leq 6n^3$ (\cref{ob:lz78-answer-string-length}).
    We then have $\sigma = 3n = \Omega(N^{1/3})$.

  \item In this step, we compute the product $AB$. By
    \cref{lm:sym-occ-and-mat-mul-equivalence},
    computing $(AB)[i,j]$ for any given $i,j \in [1 \dd n]$
    can be done using a single symbol occurrence query on the string $T_{A,B,i}$ (\cref{def:mat-mul-answer-string}). More precisely,
    by \cref{lm:sym-occ-and-mat-mul-equivalence}, it holds
    $(AB)[i,j] = \Occurs{T_{A,B,i}}{0}{|T_{A,B,i}|}{j}$. On the other hand, note that by
    comparing \cref{def:mat-mul-answer-string} and \cref{def:mat-mul-lz78-answer-string}, we immediately
    have that since $j \in [1 \dd n]$, it holds $\Occurs{T_{A,B,i}}{0}{|T_{A,B,i}|}{j} = \Occurs{X_{A,B,i}}{0}{|X_{A,B,i}|}{j}$.
    Putting this together, we thus obtain that
    \begin{align*}
      (AB)[i,j]
        &= \Occurs{T_{A,B,i}}{0}{|T_{A,B,i}|}{j}\\
        &= \Occurs{X_{A,B,i}}{0}{|X_{A,B,i}|}{j}\\
        &= \Occurs{X_{A,B}}{\Delta+A_{\rm sum}[i-1]}{\Delta+A_{\rm sum}[i]}{j}\\
        &= \Occurs{T}{\Delta+A_{\rm sum}[i-1]}{\Delta+A_{\rm sum}[i]}{j}.
    \end{align*}
    Using the array $A_{\rm sum}$, in $\bigO(n^2)$ time we prepare
    the arguments for a batch of $m = n^2$ symbol occurrence queries on $T$.
    Using the algorithm from the claim, we then answer all the queries
    (and hence compute the product $AB$) in
    $
      \bigO((\BigLZSize{T} + m) \log^{\bigO(1)} N) =
      \bigO(n^2 \log^{\bigO(1)} (n^3)) =
      \bigO(n^2 \log^{\bigO(1)} n)
    $
    time.
 \end{enumerate}

 In total, the algorithm takes $\bigO(n^2 \log^{\bigO(1)} n)$ time.
\end{proof}

%% file: hardness-of-approx-rank.tex
\subsection{Hardness of Additive Approximation of Rank on Grammars}\label{sec:hardness-of-approx-rank-queries}

\begin{definition}[Stretch operation]\label{def:strword}
  For every $T \in \Sigma^{n}$ and $k \in \Zp$, we denote
  \[
    \strword{k}{T} := \textstyle\bigodot_{i=1,\dots,n} T[i]^k \in \Sigma^{n \cdot k}.
  \]
\end{definition}

\begin{lemma}\label{lm:stretch}
  Let $T \in \Sigma^{n}$, $k \in \Zp$, and
  $T' = \strword{k}{T}$ (\cref{def:strword}).
  For every $b, e \in [0 \dd n]$ and every $c \in \Sigma$,
  the following two statements hold
  (see \cref{def:symbol-occ,def:rank}):
  \begin{enumerate}
  \item If $\Occurs{T}{b}{e}{c} = 0$, then $\TwoSidedRank{T'}{kb}{ke}{c} = 0$.
  \item If $\Occurs{T}{b}{e}{c} = 1$, then $\TwoSidedRank{T'}{kb}{ke}{c} \geq k$.
  \end{enumerate}
\end{lemma}
\begin{proof}
  If $b \geq e$, then we have $kb \geq ke$,
  $\Occurs{T}{b}{e}{c} = 0$, and $\TwoSidedRank{T'}{kb}{ke}{c} = 0$. Thus, the
  claim in this case holds. Let us now assume $b < e$.
  By \cref{def:strword}, $T'(kb \dd ke] = \bigodot_{i=b+1,\dots,e} T[i]^{k}$.
  Thus, if $c$ does not occur in $T(b \dd e]$, then it also
  does not occur in $T'(kb \dd ke]$, and hence
  $\TwoSidedRank{T'}{kb}{ke}{c} = 0$. If $c$ occurs in $T(b \dd e]$ then,
  letting $i \in (b \dd e]$ be such that $T[i] = c$, $T[i]^{k}$ occurs
  in $T'(kb \dd ke]$, and hence $\TwoSidedRank{T'}{kb}{ke}{c} \geq k$.
\end{proof}

\begin{proposition}\label{pr:slp-stretching}
  Let $\sigma \in \Zp$ and $\Sigma = \{1, \dots, \sigma\}$.
  Let $G = (V, \Sigma, R, S)$ be an SLP generating a nonempty string
  $T$ and let $k \in \Zp$.
  Given $G$, we can in $\bigO(|G| + \sigma \log k)$ time construct an SLG 
  $G' = (V', \Sigma, R', S')$ such that $\Lang{G'} = \{\strword{k}{T}\}$
  (\cref{def:strword}).
\end{proposition}
\begin{proof}

  Denote $V = \{N_1, \dots, N_g\}$ and let $s \in [1 \dd g]$ be such
  that $S = N_s$. Denote $h = \lfloor \log k \rfloor$, and let
  $p_1,\ldots,p_{h'} \in \Zn$ be the unique nonempty increasing sequence
  satisfying $k = \sum_{i=1}^{h'} 2^{p_i}$. We then let
  $V' =
    \{X_i\}_{i \in [1 \dd g]} \cup
    \{Y_c\}_{c \in [1 \dd \sigma]} \cup
    \{Z_{c,j}\}_{c \in [1 \dd \sigma],j \in [0 \dd h]}$
  and $S' = X_{s}$, where the definitions of rules in $V'$ are as follows:
  \begin{itemize}
  \item To define $\Rhs{G'}{X_i}$, where $i \in [1 \dd g]$, we consider two cases.
    If $|\Rhs{G}{N_i}| = 2$, then, letting $j_1, j_2 \in [1 \dd g]$ be such that
    $\Rhs{G}{N_i} = N_{j_1} N_{j_2}$, we define $\Rhs{G'}{X_i} = X_{j_1} X_{j_2}$.
    Otherwise, letting $c \in [1 \dd \sigma]$ be such
    that $\Rhs{G}{N_i} = c$, we let $\Rhs{G'}{X_i} = Y_c$.
  \item For every $c \in [1 \dd \sigma]$, we define
    $\Rhs{G'}{Y_c} = \bigodot_{t=1,\ldots,h'} Z_{c,p_t}$.
  \item To define $\Rhs{G'}{Z_{c,j}}$, where $c \in [1 \dd \sigma]$ and $j \in [0 \dd h]$,
    we again consider two cases.
    If $j = 0$, then we let $\Rhs{G'}{Z_{c,j}} = c$.
    Otherwise, we let $\Rhs{G'}{Z_{c,j}} = Z_{c,j-1} Z_{c,j-1}$.
  \end{itemize}

  The size of $G'$ is $\bigO(g + \sigma h' + \sigma h) = \bigO(|G| + \sigma \log k)$.
  Given $G$, we can easily construct $G'$ in $\bigO(|G| + \sigma \log k)$ time.

  To show that $\Lang{G'} = \{\strword{k}{T}\}$, it suffices to observe that,
  for every $c \in [1 \dd \sigma]$ and $j \in [0 \dd h]$, it holds
  $\Exp{G'}{Z_{c,j}} = c^{2^j}$. This implies that, for every $c \in [1 \dd \sigma]$,
  $\Exp{G'}{Y_c}
    = \Exp{G'}{\bigodot_{j=1,\ldots,h'} Z_{c,p_j}}
    = \bigodot_{j=1,\ldots,h'} \Exp{G'}{Z_{c,p_j}}
    = \bigodot_{j=1,\ldots,h'} c^{2^{p_j}}
    = c^{\sum_{j=1}^{h'}2^{p_j}}
    = c^k$.
  Consequently, for
  every $i \in [1 \dd g]$, $\Exp{G'}{X_i} = \strword{k}{\Exp{G}{N_i}}$.
  In particular, $\Exp{G'}{S'} = \Exp{G'}{X_s} = \strword{k}{\Exp{G}{N_s}}
  = \strword{k}{\Exp{G}{S}} = \strword{k}{T}$, and hence we obtain
  $\Lang{G'} = \{\Exp{G'}{S'}\} = \{\strword{k}{T}\}$.
\end{proof}

\begin{proposition}\label{pr:slg-stretching}
  Let $\sigma \in \Zp$ and $\Sigma = \{1, \dots, \sigma\}$.
  Let $G = (V, \Sigma, R, S)$ be an SLG generating a nonempty string
  $T$ and let $k \in \Zp$.
  Given $G$, we can in $\bigO(|G| + \sigma \log k)$ time construct an SLG 
  $G' = (V', \Sigma, R', S')$ such that $\Lang{G'} = \{\strword{k}{T}\}$
  (\cref{def:strword}).
\end{proposition}
\begin{proof}
  The result follows by combining \cref{ob:slg-to-slp} and \cref{pr:slp-stretching}.
\end{proof}

\thhardnessofrankapproximation*
\begin{proof}
  Let $A,B \in \BinaryAlphabet^{n \times n}$ be two given Boolean matrices.
  The algorithm to compute the Boolean matrix product $AB$ proceeds as follows:
  \begin{enumerate}

  \item We compute an array $A_{\rm sum}[0 \dd n]$
    defined so that $A_{\rm sum}[0] = 0$ and, for every $i \in [1 \dd n]$,
    $A_{\rm sum}[i] = \sum_{j=1}^{i} |T_{A,B,j}|$ (where $T_{A,B,j}$ is as in
    \cref{def:mat-mul-answer-string}). Using the algorithm presented
    in the proof of \cref{th:hardness-of-slg-symbol-occ}, the computation
    of $A_{\rm sum}$ takes $\bigO(n^2)$ time. Note that, for every
    $i \in [1 \dd n]$, we then have
    $T_{A,B}(A_{\rm sum}[i-1] \dd A_{\rm sum}[i]] = T_{A,B,i}$
    (see \cref{def:mat-mul-answer-string}).

  \item We check if $A_{\rm sum}[n] = 0$. If so, then by \cref{def:mat-mul-answer-string},
    it holds $|T_{A,B}| = 0$. In this case, by \cref{lm:sym-occ-and-mat-mul-equivalence},
    the output matrix $AB$ consists only of zeros. Thus, in this case we return the
    $n \times n$ matrix consisting of zeros in $\bigO(n^2)$ time, and conclude the algorithm.
    Henceforth, we assume that $A_{\rm sum}[n] > 0$, i.e., $|T_{A,B}| > 0$.

  \item Using \cref{pr:mat-mul-answer-string-slg-construction},
    in $\bigO(n^2)$ time we construct an SLG
    $G_{A,B} = (V_{A,B}, \Sigma, R_{A,B}, S_{A,B})$ such that
    $\Sigma = \{1, \dots, \sigma\}$ (where $\sigma = n$) and
    $\Lang{G_{A,B}} = \{T_{A,B}\}$ (\cref{def:mat-mul-answer-string}).
    Note that the upper bound on the runtime of
    \cref{pr:mat-mul-answer-string-slg-construction} implies
    that $|G_{A,B}| = \bigO(n^2)$. Note also that $|T_{A,B}| \leq n^3$
    (\cref{ob:mat-mul-answer-string-length}).

  \item Apply \cref{pr:slg-stretching} to the SLG $G_{A,B}$ with $k = \lfloor 2^{1/\mu} \cdot n^{3(1/\mu-1)} \rfloor$
    (note that we can apply \cref{pr:slg-stretching} here, since $T_{A,B} \neq \emptystring$).
    This takes $\bigO(|G_{A,B}| + \sigma \log k) = \bigO(n^2)$ time,
    and we obtain an SLG $G = (V, \Sigma, R, S)$ such that $\Lang{G} = \{T\}$,
    where $T = \strword{k}{T_{A,B}}$ (\cref{def:strword}).
    The upper bound on the runtime of
    \cref{pr:slg-stretching} implies that $|G| = \bigO(n^2)$.
    Denote $N = |T|$ and note that
    \[
      N = k \cdot |T_{A,B}| \leq \lfloor 2^{1/\mu} \cdot n^{3(1/\mu-1)} \rfloor \cdot n^3 \leq 2^{1/\mu} \cdot n^{3/\mu}.
    \]
    Since $\sigma = n$, we thus obtain $N \leq 2^{1/\mu} \cdot \sigma^{3/\mu}$, i.e., $\sigma = \Omega(N^{\mu/3})$.
    Thus, all the conditions for applying the algorithm from the claim hold for
    SLG $G$.

  \item In this step, we compute the product $AB$.
    We first make the following observations:
    \begin{itemize}
    \item Note that
      $N^{1-\mu} \leq (2^{1/\mu} \cdot n^{3/\mu})^{1-\mu} = 2^{1/\mu-1} \cdot n^{3(1/\mu-1)}$. This implies
      $\lfloor N^{1-\mu} \rfloor \leq \lfloor 2^{1/\mu-1} \cdot n^{3(1/\mu-1)} \rfloor \leq \tfrac{1}{2} \lfloor 2^{1/\mu} \cdot n^{3(1/\mu-1)} \rfloor = k/2$.
    \item Consider
      any $c \in \Sigma$ and any $b,e \in [0 \dd |T_{A,B}|]$.
      Let $x$ be a $\lfloor N^{1-\mu} \rfloor$-additive approximation of
      $\TwoSidedRank{T}{kb}{ke}{c}$. Observe that:
      \begin{itemize}
      \item If $\Occurs{T_{A,B}}{b}{e}{c} = 0$, then since
        $T = \strword{k}{T_{A,B}}$, it follows by
        \cref{lm:stretch} that
        $\TwoSidedRank{T}{kb}{ke}{c} = 0$. Thus,
        by \cref{def:additive-approximation},
        $x < \lfloor N^{1-\mu} \rfloor \leq k/2$.
      \item If $\Occurs{T_{A,B}}{b}{e}{c} = 1$, then by \cref{lm:stretch}, it holds
        $\TwoSidedRank{T}{kb}{ke}{c} \geq k$. Thus,
        by \cref{def:additive-approximation},
        we have $x > \TwoSidedRank{T}{kb}{ke}{c} - \lfloor N^{1-\mu} \rfloor \geq k/2$.
      \end{itemize}
      Consequently, we can determine $\Occurs{T_{A,B}}{b}{e}{c}$
      from $x$ in $\bigO(1)$ time.
    \end{itemize}
    By \cref{lm:sym-occ-and-mat-mul-equivalence}, for any $i,j \in [1 \dd n]$,
    it holds $(AB)[i,j] = \Occurs{T_{A,B}}{A_{\rm sum}[i-1]}{A_{\rm sum}[i]}{j}$.
    Consequently, to compute the product $AB$, we
    compute the $\lfloor N^{1-\mu} \rfloor$-additive
    approximation of the value $\TwoSidedRank{T}{k\cdot A_{\rm sum}[i-1]}{k\cdot A_{\rm sum}[i]}{j}$
    for every $i,j \in [1 \dd n]$. By the above discussion, this
    lets us compute all the values in $AB$.
    Answering the batch of $m = n^2$
    approximate rank queries takes
    $\bigO((|G| + m) \log^{\bigO(1)} N) = \bigO(n^2 \log^{\bigO(1)} (n^{3/\mu}))
    = \bigO(n^2 \log^{\bigO(1)} n)$ time.
	\end{enumerate}

	In total, the algorithm takes $\bigO(n^2 \log^{\bigO(1)} n)$ time.
\end{proof}

%% file: reductions-from-ov.tex
\section{Reductions from Orthogonal Vectors}\label{sec:reductions-from-ov}

\input{hardness-of-range-distinct-counting}

\input{hardness-of-range-mode-freq}

%% file: hardness-of-range-distinct-counting.tex
\subsection{Hardness of Range Distinct Count Queries on Grammars}\label{sec:hardness-of-range-distinct-count-on-slg}

\begin{definition}[String of vector IDs with $\one$ at given coordinate]\label{def:one-vectors}
  Let $A = (a_1, \dots, a_n)$ be a sequence of $n \geq 1$ binary vectors of dimension $d \geq 1$, i.e., such that, for every
  $i \in [1 \dd n]$, it holds $a_i \in \BinaryAlphabet^{d}$.
  For every $j \in [1 \dd d]$, by $\OneVectors{A}{j}$ we denote a string containing in
  increasing order the indices of all vectors from $A$ with a one at the $j$th coordinate, i.e.,
  $\OneVectors{A}{j} \in \{1, \dots, n\}^*$ is such that, letting $k$ be its length, it holds
  \begin{itemize}
  \item if $k > 0$, then $\OneVectors{A}{j}[1] < \dots < \OneVectors{A}{j}[k]$, and
  \item $\{\OneVectors{A}{j}[t]\}_{t \in [1 \dd k]} = \{i \in [1 \dd n] : a_i[j] = \one\}$.
  \end{itemize}
\end{definition}

\begin{definition}[String of $\one$-bit positions]\label{def:one-pos}
  For any vector $x \in \BinaryAlphabet^{d}$, where $d \geq 1$, by $\OnePos{x}$ we denote a string
  containing in increasing order the coordinates of all ones in $x$, i.e.,
  $\OnePos{x} \in \{1,\dots,d\}^*$ is such that, letting $k$ be its length, it holds
  \begin{itemize}
  \item if $k > 0$, then $\OnePos{x}[1] < \dots < \OnePos{x}[k]$, and
  \item $\{\OnePos{x}[t]\}_{t \in [1 \dd k]} = \{j \in [1 \dd d] : x[j] = \one\}$.
  \end{itemize}
\end{definition}

\begin{definition}[Grammar-compressible encoding of OV orthogonality as a string]\label{def:range-distinct-count-answer-string}
  Let $A = (a_1, \dots, a_n)$ be a sequence of $n \geq 1$ binary vectors of dimension $d \geq 1$, i.e., such that,
  for every $i \in [1 \dd n]$, it holds $a_i \in \BinaryAlphabet^{d}$.
  For every $i \in [1 \dd n]$, we define (see \cref{def:one-vectors})
  \[
    U_{A,i} := \textstyle\bigodot_{j=1,\dots,k}\OneVectors{A}{R[j]},
  \]
  where $R = \OnePos{a_i}$ (\cref{def:one-pos}) and $k = |R|$. We then let
  \[
    U_A := \textstyle\bigodot_{i=1,\dots,n} U_{A,i}.
  \]
\end{definition}

\begin{observation}\label{ob:range-distinct-count-answer-string-length}
  For every sequence $A = (a_1, \dots, a_n)$ of $n \geq 1$ binary vectors of dimension $d \geq 1$,
  the string $U_A$ (\cref{def:range-distinct-count-answer-string}) satisfies $|U_A| \leq n^2 d$.
\end{observation}
\begin{proof}
  For every $j \in [1 \dd d]$, it holds $|\OneVectors{A}{j}| \leq n$ (\cref{def:one-vectors}).
  On the other hand, for every $i \in [1 \dd n]$, $|\OnePos{a_i}| \leq d$ (\cref{def:one-pos}).
  Consequently, for every $i \in [1 \dd n]$, we have $|U_{A,i}| \leq n d$ and hence we obtain that
  $|U_A| \leq n^2 d$ (see \cref{def:range-distinct-count-answer-string}).
\end{proof}

\begin{lemma}\label{lm:reduce-ov-to-range-distinct-count}
  Let $A = (a_1, \dots, a_n)$ be a sequence of $n \geq 1$ binary vectors of dimension $d \geq 1$.
  For every $i,j \in [1 \dd n]$, the following two conditions are equivalent:
  \begin{enumerate}
  \item\label{lm:reduce-ov-to-range-distinct-count-it-1}
    $\DotProduct{a_i}{a_j} \neq 0$.
  \item\label{lm:reduce-ov-to-range-distinct-count-it-2}
    The symbol $j$ occurs in the string $U_{A,i}$ (\cref{def:range-distinct-count-answer-string}).
  \end{enumerate}
\end{lemma}
\begin{proof}

  (\ref{lm:reduce-ov-to-range-distinct-count-it-1}) $\Rightarrow$
  (\ref{lm:reduce-ov-to-range-distinct-count-it-2})
  Assume that $\DotProduct{a_i}{a_j} \neq 0$.
  Then, there exists $p \in [1 \dd d]$ such that $a_i[p] = \one$ and $a_j[p] = \one$.
  By \cref{def:one-vectors}, this implies that the symbol $j$ occurs in the string $\OneVectors{A}{p}$.
  On the other hand, symbol $p$ occurs in $\OnePos{a_i}$ (\cref{def:one-pos}).
  By \cref{def:range-distinct-count-answer-string}, this implies that $\OneVectors{A}{p}$ is a substring of $U_{A,i}$.
  Consequently, the symbol $j$ occurs in the string $U_{A,i}$.

  (\ref{lm:reduce-ov-to-range-distinct-count-it-2}) $\Rightarrow$
  (\ref{lm:reduce-ov-to-range-distinct-count-it-1})
  Assume now that the symbol $j$ occurs in the string $U_{A,i}$. Denote $R = \OnePos{a_i}$ and $k = |R|$.
  By \cref{def:range-distinct-count-answer-string}, the assumption that $j$ occurs in $U_{A,i}$ implies
  that there exists $t \in [1 \dd k]$ such that $j$ occurs in $\OneVectors{A}{R[t]}$. By \cref{def:one-vectors},
  this implies that $a_j[R[t]] = \one$. On the other hand, by \cref{def:one-pos}, we have $a_i[R[t]] = \one$.
  Thus, we obtain $\DotProduct{a_i}{a_j} \neq 0$.
\end{proof}

\begin{lemma}\label{lm:ov-and-range-distinct-count-equivalence}
  Let $A = (a_1, \dots, a_n)$ be a sequence of $n \geq 1$ binary vectors of
  dimension $d \geq 1$. The following two conditions are equivalent:
  \begin{enumerate}
  \item There exist $i,j \in [1 \dd n]$ such that $\DotProduct{a_i}{a_j} = 0$.
  \item There exists $i \in [1 \dd n]$ such that $\RangeDistinctCount{U_{A,i}}{0}{|U_{A,i}|} < n$
    (see \cref{def:range-distinct-count,def:range-distinct-count-answer-string}).
  \end{enumerate}
\end{lemma}
\begin{proof}
  The equivalence follows immediately by observing that by
  \cref{lm:reduce-ov-to-range-distinct-count}, for every $i \in [1 \dd n]$,
  it holds
  \[
    \RangeDistinctCount{U_{A,i}}{0}{|U_{A,i}|}
      = |\{j \in [1 \dd n] : \DotProduct{a_i}{a_j} \neq 0\}|.
    \qedhere
  \]
\end{proof}

\begin{proposition}\label{pr:range-distinct-count-answer-string-slg-construction}
  Given any sequence $A = (a_1, \dots, a_n)$ of $n \geq 1$ binary vectors of
  dimension $d \geq 1$, we can in $\bigO(dn)$ time construct an SLG
  $G = (V, \Sigma, R, S)$ such that $\Sigma = \{1, \dots, n\}$ and
  $\Lang{G} = \{U_A\}$ (\cref{def:range-distinct-count-answer-string}).
\end{proposition}
\begin{proof}

  Let $V = \{X_1, \ldots, X_d, Y_1, \ldots, Y_n, S\}$.
  For $j \in [1 \dd d]$, we set
  $\Rhs{G}{X_j} = \OneVectors{A}{j} \in \{1, \dots, n\}^{*}$ (\cref{def:one-vectors}).
  For every $i \in [1 \dd n]$, letting $R = \OnePos{a_i}$ (\cref{def:one-pos}) and $k = |R|$,
  we set $\Rhs{G}{Y_i} = X_{R[1]} \cdot X_{R[2]} \cdots X_{R[k]}$,
  with the right-hand side interpreted as $\emptystring$ if $k=0$.
  Finally, we set $\Rhs{G}{S} = Y_{1} \cdot Y_{2} \cdots Y_{n}$.

  To construct $G$, we proceed as follows:
  \begin{enumerate}
  \item In $\bigO(dn)$ time we compute the strings $\OneVectors{A}{j}$ for all $j \in [1 \dd d]$.
  \item In $\bigO(dn)$ time we compute the strings $\OnePos{a_i}$ for all $i \in [1 \dd n]$.
  \item Given the above strings, the construction of $\Rhs{G}{N}$ for all $N \in V$
    takes $\bigO(dn)$ time.
  \end{enumerate}
  In total, the construction takes $\bigO(dn)$ time.

  To show that $\Lang{G} = \{U_A\}$, it suffices to observe that, for every
  $i \in [1 \dd n]$, it holds $\Exp{G}{Y_i} = U_{A,i}$ (\cref{def:range-distinct-count-answer-string}).
  Thus, $\Lang{G} = \{\Exp{G}{S}\} = \{\bigodot_{i=1,\ldots,n} \Exp{G}{Y_i}\} = \{U_{A}\}$.
\end{proof}

\thhardnessofrangedistinctcount*
\begin{proof}

  Suppose that there exists a constant $\epsilon > 0$ such that,
  given any SLG $G = (V, \Sigma, R, S)$ representing a string
  $T \in \Sigma^{N}$ (where $\Sigma = \{1, \ldots, \sigma\}$ and
  $\sigma = \Omega((N / \log N)^{1/2})$), we can
  answer any batch of $m = \Omega(|G| / \log N)$ range distinct counting queries in
  $\bigO(m|G|^{1-\epsilon} \log^{c'} N)$ time,
  where $c' > 0$ is some constant.
  We will show that this implies that \cref{con:ov} does not hold.

  Denote $\epsilon' = \epsilon/2$.
  Consider any sequence $A = (a_1, \ldots, a_n)$ of $n \geq 1$ binary vectors
  of dimension $d = c\log n$, where $c \geq 1$ is a constant.
  Given the sequence $A$, we determine whether there exist
  $i,j \in [1 \dd n]$ such that $\DotProduct{a_i}{a_j} = 0$ as follows:
  \begin{enumerate}

  \item In $\bigO(dn)$ time we check if there
    exists $i \in [1 \dd n]$ such that $a_i = \zero^{d}$. If we find
    such $i$, we immediately return that the given OV instance has
    a pair $i,j \in [1 \dd n]$ satisfying $\DotProduct{a_i}{a_j} = 0$,
    and the algorithm is complete. Let us now assume that no such $i$
    was found. Denote $p = \sum_{i=1}^{n} \sum_{j=1}^{d} a_i[j]$.
    Since $A$ does not contain the zero vector, we have $p \geq n$.
    Note that by \cref{def:range-distinct-count-answer-string}, we then
    obtain $|U_{A}| \geq p \geq n$.

  \item Next, we compute an array $A_{\rm sum}[0 \dd n]$
    defined such that $A_{\rm sum}[0] = 0$ and, for every $i \in [1 \dd n]$,
    $A_{\rm sum}[i] = \sum_{j=1}^{i} |U_{A,j}|$
    (see \cref{def:range-distinct-count-answer-string}). We first
    compute an array $A_{\rm len}[1 \dd n]$ defined by
    $A_{\rm len}[i] = |U_{A,i}|$. To this end, we first in $\bigO(dn)$ time
    compute the strings $\OneVectors{A}{j}$ for all $j \in [1 \dd d]$
    and the strings $\OnePos{a_i}$ for all $i \in [1 \dd n]$. Using
    these strings, we can then compute $|U_{A,i}|$ for every
    $i \in [1 \dd n]$ in $\bigO(d)$ time. In total, computing
    $A_{\rm len}$ takes $\bigO(dn)$ time. Using $A_{\rm len}$, we
    can obtain $A_{\rm sum}$ in $\bigO(n)$ time. In total, construction
    of $A_{\rm sum}$ takes $\bigO(dn)$ time.
    Observe that, by
    \cref{def:range-distinct-count-answer-string}, for any $i \in [1 \dd n]$,
    we then have
    \[
      U_{A}(A_{\rm sum}[i-1] \dd A_{\rm sum}[i]] = U_{A,i}.
    \]

  \item Using \cref{pr:range-distinct-count-answer-string-slg-construction},
    in $\bigO(dn)$ time we construct an SLG
    $G = (V, \Sigma, R, S)$ such that $\Sigma = \{1, \dots, \sigma\}$ (where $\sigma = n$)
    and $\Lang{G} = \{T\}$, where $T = U_{A}$ (\cref{def:range-distinct-count-answer-string}).
    Note that the upper bound on the runtime of \cref{pr:range-distinct-count-answer-string-slg-construction}
    implies that $|G| = \bigO(dn)$. Denote $N = |T|$ and note
    that by \cref{ob:range-distinct-count-answer-string-length}, it holds
    $N \leq dn^2$. Recall that above we also noted that $N \geq n$ (which, together with $N \leq dn^2$, implies that $\log N = \Theta(\log n)$).
    Since $\sigma = n$, we thus obtain $N \leq d\sigma^2$,
    i.e., $\sigma = \Omega((N/d)^{1/2}) = \Omega((N/\log n)^{1/2}) = \Omega((N/\log N)^{1/2})$.
    Note also that by $|G| = \bigO(dn)$, it follows that $n = \Omega(|G| / d) = \Omega(|G| / \log n) = \Omega(|G| / \log N)$.

  \item In this step, we determine whether there exist
    $i,j \in [1 \dd n]$ such that $\DotProduct{a_i}{a_j} = 0$.
    By \cref{lm:ov-and-range-distinct-count-equivalence}, this reduces to checking whether there
    exists $i \in [1 \dd n]$ such that $\RangeDistinctCount{U_{A,i}}{0}{|U_{A,i}|} < n$
    (see \cref{def:range-distinct-count,def:range-distinct-count-answer-string}).
    We thus proceed as follows.
    Using the array $A_{\rm sum}$, in $\bigO(n)$ time, we prepare arguments for a batch of
    $m = n = \Omega(|G| / \log N)$ range distinct counting queries on $T$. More precisely, the $i$th query
    is to compute $\RangeDistinctCount{T}{A_{\rm sum}[i-1]}{A_{\rm sum}[i]}$. We return
    yes if and only if at least one returned value is smaller than $n$. Answering
    this batch of queries takes
    \begin{align*}
      \bigO(m|G|^{1-\epsilon} \log^{c'} N)
        \subseteq   \bigO(n \cdot (dn)^{1-\epsilon} \log^{c'} n)
        =           \bigO(n^{2-\epsilon} \cdot \log^{1+c'-\epsilon} n)
        \subseteq   \bigO(n^{2-\epsilon'})
    \end{align*}
    time.
  \end{enumerate}

  In total, the algorithm takes $\bigO(dn + n^{2-\epsilon'}) = \bigO(n^{2-\epsilon'})$ time.
  Note that since our choice of $\epsilon'$ works for all $c$, we have thus proved
  that \cref{con:ov} does not hold.
\end{proof}

%% file: hardness-of-range-mode-freq.tex
\subsection{Hardness of Range Mode Frequency Queries on Grammars}\label{sec:hardness-of-range-mode-freq-queries}

\begin{definition}[String of vector IDs with $\zero$ at given coordinate]\label{def:zero-vectors}
  Let $A = (a_1, \dots, a_n)$ be a sequence of $n \geq 1$ binary vectors of dimension $d \geq 1$, i.e., such that, for every
  $i \in [1 \dd n]$, it holds $a_i \in \BinaryAlphabet^{d}$.
  For every $j \in [1 \dd d]$, by $\ZeroVectors{A}{j}$ we denote a string containing in
  increasing order the indices of all vectors from $A$ with a zero at the $j$th coordinate, i.e.,
  $\ZeroVectors{A}{j} \in \{1, \dots, n\}^*$ is such that, letting $k$ be its length, it holds
  \begin{itemize}
  \item if $k > 0$, then $\ZeroVectors{A}{j}[1] < \dots < \ZeroVectors{A}{j}[k]$, and
  \item $\{\ZeroVectors{A}{j}[t]\}_{t \in [1 \dd k]} = \{i \in [1 \dd n] : a_i[j] = \zero\}$.
  \end{itemize}
\end{definition}

\begin{definition}[Grammar-compressible encoding of OV orthogonality as a string]\label{def:range-mode-freq-answer-string}
  Let $A = (a_1, \dots, a_n)$ be a sequence of $n \geq 1$ binary vectors of dimension $d \geq 1$, i.e., such that,
  for every $i \in [1 \dd n]$, it holds $a_i \in \BinaryAlphabet^{d}$.
  For every $i \in [1 \dd n]$, we define (see \cref{def:zero-vectors})
  \[
    W_{A,i} := \textstyle\bigodot_{j=1,\dots,k}\ZeroVectors{A}{R[j]},
  \]
  where $R = \OnePos{a_i}$ (\cref{def:one-pos}) and $k = |R|$. We then let
  \[
    W_A := \textstyle\bigodot_{i=1,\dots,n} W_{A,i}.
  \]
\end{definition}

\begin{observation}\label{ob:range-mode-freq-answer-string-length}
  For every sequence $A = (a_1, \dots, a_n)$ of $n \geq 1$ binary vectors
  of dimension $d \geq 1$, the string $W_A$
  (\cref{def:range-mode-freq-answer-string}) satisfies $|W_A| \leq n^2d$.
\end{observation}
\begin{proof}
  The proof is analogous to the proof of \cref{ob:range-distinct-count-answer-string-length}.
\end{proof}

\begin{lemma}\label{lm:reduce-ov-to-range-mode-freq}
  Let $A = (a_1, \ldots, a_n)$ be a sequence of $n \geq 1$ binary vectors of dimension $d \geq 1$.
  For every $i, j \in [1 \dd n]$, the following conditions are equivalent:
  \begin{enumerate}
  \item\label{lm:reduce-ov-to-range-mode-freq-it-1}
    $\DotProduct{a_i}{a_j} \neq 0$.
  \item\label{lm:reduce-ov-to-range-mode-freq-it-2}
    Symbol $j$ occurs in the string $W_{A,i}$ (\cref{def:range-mode-freq-answer-string}) less than $|\OnePos{a_i}|$ times.
  \end{enumerate}
\end{lemma}
\begin{proof}

  We first establish an auxiliary property.
  Observe that, for every $p \in [1 \dd d]$, all symbols in the string $\ZeroVectors{A}{p}$ (\cref{def:zero-vectors}) are distinct.
  By \cref{def:range-mode-freq-answer-string}, this implies that, for every $i, j \in [1 \dd n]$,
  the symbol $j$ occurs in the string $W_{A,i}$ at most $|\OnePos{a_i}|$ times. Consequently, we obtain that the symbol $j$ appears in the string $W_{A,i}$
  $|\OnePos{a_i}|$ times if and only if, letting $R = \OnePos{a_i}$ and $k = |R|$, the symbol $j$ appears in the string $\ZeroVectors{A}{R[t]}$ for all $t \in [1 \dd k]$.

  (\ref{lm:reduce-ov-to-range-mode-freq-it-1}) $\Rightarrow$
  (\ref{lm:reduce-ov-to-range-mode-freq-it-2})
  Assume that $\DotProduct{a_i}{a_j} \neq 0$.
  Then, there exists $p \in [1 \dd d]$ such that $a_i[p] = \one$ and $a_j[p] = \one$.
  By \cref{def:one-pos}, $p$ appears in $\OnePos{a_i}$, i.e., letting $R = \OnePos{a_i}$ and $k = |R|$, it holds $R[t] = p$ for some $t \in [1 \dd k]$.
  On the other hand, $a_j[p] = \one$ implies, by \cref{def:zero-vectors}, that $j$ does not occur in $\ZeroVectors{A}{p} = \ZeroVectors{A}{R[t]}$.
  By the above auxiliary property, we thus obtain that $j$ appears in $W_{A,i}$ less than $|\OnePos{a_i}|$ times.

  (\ref{lm:reduce-ov-to-range-mode-freq-it-2}) $\Rightarrow$
  (\ref{lm:reduce-ov-to-range-mode-freq-it-1})
  Assume that the symbol $j$ occurs in $W_{A,i}$ less than $|\OnePos{a_i}|$ times. By the above auxiliary property, this implies that,
  letting $R = \OnePos{a_i}$ and $k = |R|$, there exists $t \in [1 \dd k]$ such that $j$ does not occur in $\ZeroVectors{A}{R[t]}$.
  By \cref{def:zero-vectors}, this implies that it holds $a_{j}[R[t]] = \one$. On the other hand, note that, by $R = \OnePos{a_i}$ and \cref{def:one-pos},
  it holds $a_{i}[R[t]] = \one$. Combining this, we thus obtain that $\DotProduct{a_i}{a_j} \neq 0$.
\end{proof}

\begin{lemma}\label{lm:ov-and-range-mode-freq-equivalence}
  Let $A = (a_1, \dots, a_n)$ be a sequence of $n \geq 1$ binary vectors of
  dimension $d \geq 1$. The following two conditions are equivalent:
  \begin{enumerate}
  \item There exist $i,j \in [1 \dd n]$ such that $\DotProduct{a_i}{a_j} = 0$.
  \item There exists $i \in [1 \dd n]$ such that $\RangeModeFreq{W_{A,i}}{0}{|W_{A,i}|} \geq |\OnePos{a_i}|$
    (see \cref{def:range-mode-freq,def:range-mode-freq-answer-string}).
  \end{enumerate}
\end{lemma}
\begin{proof}
  Note that, for every $i \in [1 \dd n]$,
  the string $W_{A,i}$ (\cref{def:range-mode-freq-answer-string})
  is over alphabet $\{1, \ldots, n\}$. Thus,
  by \cref{lm:reduce-ov-to-range-mode-freq}, we obtain that
  \[
    \{j \in [1 \dd n] : j\text{ occurs in }W_{A,i}\text{ at least }|\OnePos{a_i}|\text{ times}\}
    = \{j \in [1 \dd n] : \DotProduct{a_i}{a_j} = 0\}.
    \qedhere
  \]
\end{proof}

\begin{proposition}\label{pr:range-mode-freq-answer-string-slg-construction}
  Given any sequence $A = (a_1, \dots, a_n)$ of $n \geq 1$ binary vectors of
  dimension $d \geq 1$, we can in $\bigO(dn)$ time construct an SLG
  $G = (V, \Sigma, R, S)$ such that $\Sigma = \{1, \dots, n\}$ and
  $\Lang{G} = \{W_{A}\}$ (\cref{def:range-mode-freq-answer-string}).
\end{proposition}
\begin{proof}
  The proof is analogous to the proof of \cref{pr:range-distinct-count-answer-string-slg-construction}.
\end{proof}

\thhardnessofrangemodefreq*
\begin{proof}

  Suppose that there exists a constant $\epsilon > 0$ such that,
  given any SLG $G = (V, \Sigma, R, S)$ representing a string
  $T \in \Sigma^{N}$ (where $\Sigma = \{1, \ldots, \sigma\}$ and
  $\sigma = \Omega((N / \log N)^{1/2})$), we can
  answer any batch of $m = \Omega(|G| / \log N)$ range mode frequency queries in
  $\bigO(m|G|^{1-\epsilon} \log^{c'} N)$ time,
  where $c' > 0$ is some constant.
  We will show that this implies that \cref{con:ov} does not hold.

  Denote $\epsilon' = \epsilon/2$.
  Consider any sequence $A = (a_1, \ldots, a_n)$ of $n \geq 2$ binary vectors
  of dimension $d = c\log n$, where $c \geq 1$ is a constant.
  Given the sequence $A$, we determine whether there exist
  $i,j \in [1 \dd n]$ such that $\DotProduct{a_i}{a_j} = 0$ as follows:
  \begin{enumerate}

  \item In $\bigO(dn)$ time we check whether for all
    $i \in (1 \dd n]$ it holds $a_i = a_1$. If yes, then we are
    able to solve this OV instance simply by checking if $a_1 = \zero^{d}$,
    and this check concludes the algorithm. Let us thus assume that
    there exists $i \in (1 \dd n]$ such that $a_i \neq a_1$, i.e., not
    all vectors in the input sequence are equal. Then,
    there exists $j \in [1 \dd d]$ such that for some
    $i, i' \in [1 \dd n]$ it holds $a_{i}[j] = \zero$ and $a_{i'}[j] = \one$.
    Denote $n_0 = |\{i \in [1 \dd n] : a_{i}[j] = \zero\}|$ and $n_1 = n - n_0$.
    Note that it holds $n_1 \geq 1$ and $|\ZeroVectors{A}{j}| = n_0 \geq 1$.
    Since for every $i \in [1 \dd n]$ satisfying $a_i[j] = \one$, the string
    $\ZeroVectors{A}{j}$ is a substring of $W_{A,i}$
    (see \cref{def:range-mode-freq-answer-string}),
    it follows that for every such index $i$, it holds $|W_{A,i}| \geq n_0$.
    Since there are $n_1$ such indices $i$, we thus obtain
    \[
      |W_{A}| = \textstyle\sum_{i=1}^{n} |W_{A,i}| \geq n_1 \cdot n_0.
    \]
    By $n_0 \geq 1$, $n_1 \geq 1$, and $n_0 + n_1 = n$, we thus obtain
    $|W_{A}| \geq n - 1$.

  \item In $\bigO(dn)$ time we compute an array $A_{\rm ones}[1 \dd n]$
    defined by $A_{\rm ones}[i] = |\OnePos{a_i}|$.

  \item Next, we compute an array $A_{\rm sum}[0 \dd n]$
    defined such that $A_{\rm sum}[0] = 0$ and, for every $i \in [1 \dd n]$,
    $A_{\rm sum}[i] = \sum_{j=1}^{i} |W_{A,j}|$
    (see \cref{def:range-mode-freq-answer-string}). The algorithm
    proceeds analogously as in the proof of \cref{th:hardness-of-range-distinct-count}.
    Observe that, by
    \cref{def:range-mode-freq-answer-string}, for any $i \in [1 \dd n]$,
    we then have
    \[
      W_{A}(A_{\rm sum}[i-1] \dd A_{\rm sum}[i]] = W_{A,i}.
    \]

  \item Using \cref{pr:range-mode-freq-answer-string-slg-construction},
    in $\bigO(dn)$ time we construct an SLG
    $G = (V, \Sigma, R, S)$ such that $\Sigma = \{1, \dots, \sigma\}$ (where $\sigma = n$)
    and $\Lang{G} = \{T\}$, where $T = W_{A}$ (\cref{def:range-mode-freq-answer-string}).
    Note that the upper bound on the runtime of \cref{pr:range-mode-freq-answer-string-slg-construction}
    implies that $|G| = \bigO(dn)$. Denote $N = |T|$ and note
    that by \cref{ob:range-mode-freq-answer-string-length}, it holds
    $N \leq dn^2$. Recall that above we also noted that $N \geq n-1$ (which, together with $N \leq dn^2$, implies that $\log N = \Theta(\log n)$).
    Since $\sigma = n$, we thus obtain $N \leq d\sigma^2$,
    i.e., $\sigma = \Omega((N/d)^{1/2}) = \Omega((N/\log n)^{1/2}) = \Omega((N/\log N)^{1/2})$.
    Note also that by $|G| = \bigO(dn)$, it follows that $n = \Omega(|G| / d) = \Omega(|G| / \log n) = \Omega(|G| / \log N)$.

  \item In this step, we determine whether there exist
    $i,j \in [1 \dd n]$ such that $\DotProduct{a_i}{a_j} = 0$.
    By \cref{lm:ov-and-range-mode-freq-equivalence}, this reduces to checking whether there
    exists $i \in [1 \dd n]$ such that it holds
    $\RangeModeFreq{W_{A,i}}{0}{|W_{A,i}|} \geq A_{\rm ones}[i]$
    (see \cref{def:range-mode-freq,def:range-mode-freq-answer-string}).
    We thus proceed as follows.
    Using the array $A_{\rm sum}$, in $\bigO(n)$ time, we prepare arguments for a batch of
    $m = n = \Omega(|G| / \log N)$ range mode frequency queries on $T$. More precisely, the $i$th query
    is to compute $\RangeModeFreq{T}{A_{\rm sum}[i-1]}{A_{\rm sum}[i]}$. We return
    yes if and only if at least one returned value is at least the corresponding value $A_{\rm ones}[i]$. Answering
    this batch of queries takes
    \begin{align*}
      \bigO(m|G|^{1-\epsilon} \log^{c'} N)
        \subseteq   \bigO(n \cdot (dn)^{1-\epsilon} \log^{c'} n)
        =           \bigO(n^{2-\epsilon} \cdot \log^{1+c'-\epsilon} n)
        \subseteq   \bigO(n^{2-\epsilon'})
    \end{align*}
    time.
  \end{enumerate}

  In total, the algorithm takes $\bigO(dn + n^{2-\epsilon'}) = \bigO(n^{2-\epsilon'})$ time.
  Note that since our choice of $\epsilon'$ works for all $c$, we have thus proved
  that \cref{con:ov} does not hold.
\end{proof}

%% file: appendix.tex
\appendix

\section{Appendix}

\subsection{Answering a Range Distinct Count Query on a Grammar}\label{sec:range-distinct-count-on-slg-in-linear-time}

\begin{proposition}\label{pr:range-distinct-count-on-slp-in-linear-time-small-alphabet}
  Let $G = (V, \Sigma, R, S)$ be an SLP representing a string
  $T \in \Sigma^{n}$, where $\Sigma = \{1,\ldots,\sigma\}$,
  $\sigma \leq |G|$, and $G$ is of height $h = \bigO(\log n)$.
  Given $G$ and any $b, e \in [0 \dd n]$ with $b \leq e$,
  we can compute $\RangeDistinctCount{T}{b}{e}$
  (\cref{def:range-distinct-count}) in $\bigO(|G|)$ time.
\end{proposition}
\begin{proof}

  Denote $V = \{N_1, \ldots, N_{g}\}$. Assume that the definitions
  of all nonterminals in $G$ are represented using an array
  $A_{\rm rhs}[1 \dd g]$ defined such that, for every $i \in [1 \dd g]$,
  \begin{itemize}
  \item If $|\Rhs{G}{N_i}| = 1$, then
    $A_{\rm rhs}[i] = \Rhs{G}{N_i} \in [1 \dd \sigma]$.
  \item Otherwise $A_{\rm rhs}[i] = (i', i'')$, where
    $i', i'' \in [1 \dd g]$ are such that $\Rhs{G}{N_i} = N_{i'} N_{i''}$.
  \end{itemize}

  The algorithm to compute $\RangeDistinctCount{T}{b}{e}$ proceeds as follows:
  \begin{enumerate}

  \item In the first step we construct a directed graph
    $D = (V', E')$, where $V' = \{0,1,\ldots,g+\sigma\}$ and
    the set of directed edges $E'$ is defined such that, for every
    $(u,v) \in V'^2$, $(u,v) \in E'$ holds if and only if
    $u \in [1 \dd g]$ and one of the following two conditions hold:
    \begin{itemize}
    \item $|\Rhs{G}{N_u}| = 2$ and, letting $(i',i'') = A_{\rm rhs}[u]$,
      it holds either $v = i'$ or $v = i''$,
    \item $|\Rhs{G}{N_u}| = 1$, and, letting $c = \Rhs{G}{N_u}$,
      it holds $v = g + c$.
    \end{itemize}
    Given the array $A_{\rm rhs}[1 \dd g]$, the computation
    of $D$ takes $\bigO(|G|)$ time.

  \item Next, we compute an array $A_{\rm explen}[1 \dd g]$
    defined such that, for every $i \in [1 \dd g]$,
    $A_{\rm explen}[i] = |\Exp{G}{N_i}|$.
    Given the array $A_{\rm rhs}$,
    computing $A_{\rm explen}$ takes $\bigO(|G|)$ time.

  \item We use the array $A_{\rm rhs}$ and $A_{\rm explen}$ to compute
    a sequence $(a_i)_{i \in [1 \dd k]}$ of integers in $[1 \dd g]$ that
    satisfies
    \begin{itemize}
    \item $k = \bigO(\log n)$ and
    \item $T(b \dd e] = \bigodot_{i=1,\ldots,k} \Exp{G}{N_{a_i}}$.
    \end{itemize}
    It is easy to see that such a sequence exists since $G$ has height
    $h = \bigO(\log n)$.
    After computing the sequence $(a_i)_{i \in [1 \dd k]}$, we update
    the graph $D$ by adding to $E'$ the set of edges
    $\{(0,a_i) : i \in [1 \dd k]\}$.
    This step takes $\bigO(k) = \bigO(\log n) = \bigO(|G|)$ time
    (see \cref{ob:slp-min-size}).

  \item We perform a depth-first search in $D$ from vertex $0$ marking all reachable
    vertices, and finally, count how many vertices among those in the
    subset $\{g+1,\ldots,g+\sigma\}$ were reachable, and return that
    as the answer. This step takes $\bigO(|V'| + |E'|) = \bigO(|G|)$ time.
  \end{enumerate}
  In total, the algorithm takes $\bigO(|G|)$ time.
\end{proof}

\begin{proposition}\label{pr:range-distinct-count-on-slp-in-near-linear-time}
  Let $G = (V, \Sigma, R, S)$ be an SLP representing a string
  $T \in \Sigma^{n}$, where $\Sigma$ is any ordered set
  and $G$ is of height $h = \bigO(\log n)$.
  Given $G$ and any $b, e \in [0 \dd n]$ with $b \leq e$,
  we can compute $\RangeDistinctCount{T}{b}{e}$
  (\cref{def:range-distinct-count}) in $\bigO(|G| \log n)$ time.
\end{proposition}
\begin{proof}

  Denote $V = \{N_1, \ldots, N_{g}\}$. Assume that the definitions
  of all nonterminals in $G$ are represented using an array
  $A_{\rm rhs}[1 \dd g]$ defined such that, for every $i \in [1 \dd g]$,
  \begin{itemize}
  \item If $|\Rhs{G}{N_i}| = 1$, then
    $A_{\rm rhs}[i] = \Rhs{G}{N_i} \in \Sigma$.
  \item Otherwise $A_{\rm rhs}[i] = (i', i'')$, where
    $i', i'' \in [1 \dd g]$ are such that $\Rhs{G}{N_i} = N_{i'} N_{i''}$.
  \end{itemize}
  We also denote $\Sigma_{T} = \{T[i] : i \in [1 \dd n]\}$.

  The algorithm to compute $\RangeDistinctCount{T}{b}{e}$ proceeds as follows:
  \begin{enumerate}

  \item In this step, we compute an array $A'_{\rm rhs}[1 \dd g]$
    defined such that, for every $i \in [1 \dd g]$,
    \begin{itemize}
    \item If $|\Rhs{G}{N_i}| = 1$, then
      $A'_{\rm rhs}[i] = |\{c \in \Sigma_{T} : c \preceq \Rhs{G}{N_i}\}| \in [1 \dd |\Sigma_{T}|]$.
    \item Otherwise, $A'_{\rm rhs}[i] = A_{\rm rhs}[i]$.
    \end{itemize}
    To compute the array $A'_{\rm rhs}$, we first lexicographically sort all
    the pairs in the set $\{(\Rhs{G}{N_i}, i) : i \in [1 \dd g]\text{ and }|\Rhs{G}{N_i}| = 1\}$
    (which is easy to obtain from $A_{\rm rhs}$ in $\bigO(|G|)$ time) in $\bigO(|G| \log |G|) = \bigO(|G| \log n)$ time.
    With one more scan of this sorted sequence,
    it is straightforward to compute the array $A'_{\rm rhs}$ in $\bigO(|G|)$ time.

  \item Let $G'$ denote an SLP obtained from $G$ by replacing all the rules (represented in array $A_{\rm rhs}$)
    with rules represented in $A'_{\rm rhs}$, and let $T'$ be the string represented by $G'$.
    The resulting grammar is over alphabet $\Sigma' = [1 \dd |\Sigma_{T}|]$
    which satisfies $|\Sigma'| \leq |G| = |G'|$. Since $T'$ is obtained from $T$ by a one-to-one renaming of symbols,
    $\RangeDistinctCount{T}{b}{e} = \RangeDistinctCount{T'}{b}{e}$. Thus, using
    \cref{pr:range-distinct-count-on-slp-in-linear-time-small-alphabet},
    we can compute $\RangeDistinctCount{T'}{b}{e}$ in $\bigO(|G|)$ time.

  \end{enumerate}
  In total, the algorithm takes $\bigO(|G| \log n)$ time.
\end{proof}

\begin{theorem}\label{th:range-distinct-count-on-slg-in-near-linear-time}
  Let $G = (V, \Sigma, R, S)$ be an SLG representing a string
  $T \in \Sigma^{n}$, where $n \geq 1$ and $\Sigma$ is any ordered set.
  Given $G$ and any $b, e \in [0 \dd n]$ with $b \leq e$,
  we can compute $\RangeDistinctCount{T}{b}{e}$
  (\cref{def:range-distinct-count}) in $\bigO(|G| \log n)$ time.
\end{theorem}
\begin{proof}
  The algorithm to compute $\RangeDistinctCount{T}{b}{e}$ proceeds as follows:
  \begin{enumerate}

  \item Using \cref{ob:slg-to-slp}, in $\bigO(|G|)$ time,
    compute an SLP $G' = (V', \Sigma, R', S')$
    such that $\Lang{G'} = \Lang{G} = \{T\}$.
    Note that the upper bound on the runtime of \cref{ob:slg-to-slp}
    implies that $|G'| = \bigO(|G|)$.

  \item Using \cref{lm:slp-balancing}, in $\bigO(|G|)$ time,
    compute an SLP $G'' = (V'', \Sigma, R'', S'')$
    of height $h = \bigO(\log n)$
    such that $\Lang{G''} = \Lang{G'} = \{T\}$.
    Similarly as above, note that $|G''| = \bigO(|G'|) = \bigO(|G|)$.

  \item Apply \cref{pr:range-distinct-count-on-slp-in-near-linear-time} to $G''$
    to compute $\RangeDistinctCount{T}{b}{e}$
    in $\bigO(|G''| \log n) = \bigO(|G| \log n)$ time.
  \end{enumerate}
  In total, the algorithm takes $\bigO(|G| \log n)$ time.
\end{proof}